\documentclass[11pt, reqno]{amsart}
\usepackage[a4paper,vmargin={3cm,3cm},hmargin={3cm,3cm},bottom=30mm]{geometry}
\usepackage{amsmath,amsfonts,amscd,amssymb,amsthm}
\usepackage{soul}
\usepackage{anyfontsize}
\usepackage{subcaption}
\numberwithin{equation}{section}
\usepackage{xfrac}

\newtheorem{theorem}{Theorem}[section]
\newtheorem{definition}[theorem]{Definition}

\newtheorem{lemma}[theorem]{Lemma}
\newtheorem{proposition}[theorem]{Proposition}

\newtheorem{remark}{Remark}
\newtheorem{claim}[theorem]{Claim}

\newtheorem{algorithm}{Algorithm}

\numberwithin{equation}{section}

\def\1{{\mathchoice {\rm 1\mskip-4mu l} {\rm 1\mskip-4mu l}
{\rm 1\mskip-4.5mu l} {\rm 1\mskip-5mu l}}}

\usepackage{color}

\definecolor{Blue}{cmyk}{1,1,0,0}

\definecolor{Green}{cmyk}{1,0,1,0}

\DeclareMathOperator{\var}{Var}
\DeclareMathOperator{\argmin}{argmin}
\DeclareMathOperator{\argmax}{argmax}

\usepackage[pdfpagelabels=true,unicode=true]{hyperref}
\hypersetup
{
	pdfauthor={Leandro Chiarini, Milton Jara and Wioletta M. Ruszel},
	pdftitle={Constructing fractional Gaussian fields from long-range divisible sandpiles},
	pdfkeywords={divisible sandpile, odometer, bi-Laplacian field, Gaussian field, Green's function, abstract Wiener space, long-range random walks, scaling limits},
	pdfsubject={divisible sandpile, long range, odometer, Gaussian field},
	colorlinks=true,
	linkcolor=blue,
	citecolor=blue,
	filecolor=blue,
	urlcolor=blue
}

\title[Long-range divisible sandpiles on the torus]{Constructing fractional Gaussian fields from long-range divisible sandpiles on the torus}

\author[L. Chiarini]{\small Leandro Chiarini}
\address{IMPA, Estrada Dona Castorina 110, 22460-320, Rio de Janeiro, Brazil}
\vspace{-1cm}
\address{Utrecht University, Budapestlaan 6, 3584 CD Utrecht, The Netherlands}
\email{chiarini@impa.br / l.chiarinimedeiros@uu.nl}
\author[M. Jara]{\small Milton Jara}
\address{IMPA, Estrada Dona Castorina 110, 22460-320, Rio de Janeiro, Brazil}
\email{mjara@impa.br}
\author[W. M. Ruszel]{\small Wioletta M. Ruszel}
\address{Utrecht University, Budapestlaan 6, 3584 CD Utrecht, The Netherlands}
\email{w.m.ruszel@uu.nl}

\begin{document}

\begin{abstract}
In \cite{Cipriani2016}, the authors proved that, with the appropriate rescaling,
the odometer of the (nearest neighbours) divisible sandpile on the unit torus
converges to a bi-Laplacian field. Here, we study $\alpha$-long-range divisible
sandpiles, similar to those introduced in  \cite{Frometa2018}. We show that, for
$\alpha \in (0,2)$, the limiting field is a fractional Gaussian field on the
torus with parameter $\alpha/2$.  However, for $\alpha \in [2,\infty)$, we recover the bi-Laplacian field.
This provides an alternative construction of fractional Gaussian fields such
as the Gaussian Free Field or membrane model using a diffusion based on the generator of L\'evy walks. The central
tool for obtaining our results is a careful study of the spectrum of the fractional
Laplacian on the discrete torus. More specifically, we need the rate of
divergence of the eigenvalues as we let the side length of the discrete torus
go to infinity.  As a side result, we obtain precise asymptotics for the eigenvalues of
discrete fractional Laplacians. Furthermore, we determine the order of the
expected maximum of the discrete fractional Gaussian field with parameter $\gamma=\min \{\alpha,2\}$ and $\alpha \in \mathbb{R}_+\backslash\{2\}$ on a finite grid.
\end{abstract}



\maketitle

\section{Introduction} \label{sec-intro}
The divisible sandpile model is the continuous fixed energy counterpart of the
Abelian sandpile model, which was introduced in \cite{Bak87} as a discrete toy
model displaying self-organised criticality. Self-organised critical models are
characterised by a power-law behaviour of certain quantities such as two-point correlation 
functions without fine-tuning any external parameter. The divisible sandpile model was introduced in
\cite{Levine09}.  It gives insight into the behaviour of internal diffusion limited aggregation
growth models on $\mathbb{Z}^d$ due to its similarity.

Consider a finite graph $G$ (e.g. a discrete torus $(\mathbb{Z}/n\mathbb{Z})^d$)
and initially assign randomly to each vertex a real number drawn from a given
distribution. This real number plays the role of a \textit{mass} in case the
number is positive and a \textit{hole} otherwise. At each time step, topple all
vertices with mass strictly larger than 1 by keeping mass 1 and redistributing the excess to
its neighbours. Two different redistribution types can be considered: either
redistribution of mass happens to nearest neighbours (we will call the associated model nearest neighbour divisible sandpile) 
or to all neighbours according to their relative distance to the unstable vertex
and depending on a parameter $\alpha>0$ (long-range divisible sandpile). Under certain conditions (described in
\cite{Levine2016}), the sandpile configuration will stabilise, meaning that all
heights will be equal to 1.

If we depict now the total amount of mass emitted from each vertex of the graph
upon stabilisation (odometer), we can interpret the odometer function as a
\textit{random interface model} on the discrete graph $G$. Examples of
interfaces in nature are hypersurfaces separating ice and water at $0^o$ C. Fractional Laplacians $(-\Delta)^{\sfrac{\alpha}{2}}$ describe
diffusion processes due to random displacement over long distances. Applications
in physics include turbulent fluid motions \cite{Churba2016, Epps2018} or
anomalous transport in fractured media \cite{Obe2018}. For a general reference
and more applications see also \cite{Valdi17, Poz16}. A survey about random interface models can be found in \cite{Funaki16} and about
scaling limits of odometers of divisible sandpiles on the torus in \cite{WSurvey}.

For the nearest neighbour divisible sandpile, we get the following central limit type of behaviour.
If the initial configuration satisfies a second moment and a certain independence condition, then the
rescaled odometer converges to a bi-Laplacian field in some appropriate Sobolev
space, see Theorem~2 in \cite{Cipriani2016}.

In this paper, we study the divisible sandpile model, which is redistributing its
excess mass to all the vertices of the d-dimensional torus upon each toppling. The amount of
mass emitted from $x$ and received by $y$ depends on the distance
$||x-y||^{-\alpha}$ (where $\| \cdot \|$ denotes the Euclidean norm) and is tuned by some
parameter $\alpha$, for $\alpha \in (0,\infty)$. A related problem was studied
in \cite{Frometa2018} where the authors consider a divisible sandpile model on
$\mathbb{Z}^d$ with a deterministic initial configuration, supported on a finite
domain and redistributing the excess mass according to a truncated long-range
random walk. They study the scaling limit of the odometer by exploring the connection of the limiting distribution
to an obstacle problem for a truncated fractional Laplacian.
This connection was established for the nearest neighbour divisible sandpile model in Lemma~2.2 in \cite{LevineLap}. 

The main results and novelty of the paper include determining upper and lower bounds for the expected odometer on the discrete torus for an initial Gaussian configuration for all $\alpha \neq 2$, which is stated in Theorem~\ref{thm-bounds-odometer-finite}, the scaling limit of the odometer function to a fractional Gaussian field  fGF$_{\gamma}(\mathbb{T}^d)$, $\gamma=\min\{\alpha,2\}$ and $\alpha \in
(0,\infty)$ on the continuous torus $\mathbb{T}^d$ in an appropriate Sobolev space depending on $\alpha$ in Theorem~\ref{theorem-main-non-Gaussian} and explicit asymptotics for the eigenvalues of discrete fractional Laplacians in Lemmas~\ref{lem-bounds-on-eigenvalues-1}, \ref{lem-bounds-on-eigenvalues-4} and \ref{lem-bounds-on-eigenvalues-5} for all $\alpha >0$. Note that the expected odometer is equal to the expected maximum of the discrete (massive) fractional Gaussian field on the discrete torus, when the initial configuration is Gaussian.

The structure of the proof of Theorem~\ref{thm-bounds-odometer-finite} is similar to the proof of Theorem~1.2 in \cite{Levine2016} and for the scaling limit in Theorem~\ref{theorem-main-non-Gaussian}  we rely on Theorem~2 in  \cite{Cipriani2016} proven for the nearest
neighbour case. The crucial part of the proofs involves a careful
analysis of the eigenvalues of the discrete fractional Laplacian for the different values of
$\alpha$.

In \cite{Jan} the authors constructed fractional Gaussian fields fGF$_{\gamma}(\mathbb{T}^d)$ with $\gamma\geq
2$ for correlated initial Gaussian configurations and nearest neighbour
redistribution. Note that starting initially with correlated Gaussians can only
produce fields which are in some sense \textit{smoother} than the bi-Laplacian ($\gamma=2$)
and never of Gaussian Free Field (GFF) type ($\gamma=1$) which is included in our results. 
The GFF is a very well known interface model which plays a crucial role in random field
theory, lattice statistical physics, stochastic partial differential equations
and quantum gravity theory in dimension $d=2$. 

Let us stress two interesting facts. 
Firstly,  we are constructing the
GFF on the continuous torus as a scaling limit of a discrete fractional field on the discrete torus. 

Secondly, for all $\alpha\geq 2$ our
limiting field will be the bi-Laplacian field, also known as the membrane model, which is an
important variation of the GFF. This model is becoming more studied over
the past few years from a mathematical perspective, due to its own interest
\cite{Bolthausen2017,Cipriani2018} and its connections with uniform spanning
trees \cite{Lawler2016,Sun2013}. 

Let us give some heuristics for the
change in behaviour according to $\alpha$. For $\alpha \in (0,2)$, the long-range
random walk on the torus  has a mixing
time of order at most $n^\alpha \log(n)$ versus the usual $n^2 \log(n)$ of the
simple random walk. The mixing time of the random walk is increasing in $\alpha$, we
expect the same to hold for the speed with which the sandpile configuration converges to its stable
configuration. In this case, choosing small  $\alpha$ implies that the sandpile configuration
is close to stability after fewer toppling steps, hence the short-term behaviour of the
dynamics dominates the odometer. Intuitively, each vertex $x$ emits less mass upon stabilisation and its final odometer becomes less dependent on the odometer of
 vertices far away from $x$.  As $\alpha$ increases,  the long-time
behaviour of the dynamics becomes more relevant for the odometer at each point $x$,
smoothing the effects of the initial condition since toppling happens to close neighbours of $x$.
For $\alpha >2$, the central
limit theorem guarantees that the long-term behaviour of the random walk (and
therefore the sandpile dynamics) will behave similarly to the simple random
walk. In other words, as the long-range random walk mixes faster, the odometer field
becomes less regular and has a larger expectation.

This paper is organised as follows. Section~\ref{sec-not-and-def} provides all
necessary definitions and notations. In particular, we define the long-range
divisible sandpile model, abstract Wiener spaces and introduce notations for the
Fourier analysis on the torus. The subsequent Section~\ref{sec-res} contains 
our results regarding bounds for the expected odometer (expected maximum of the discrete fractional Gaussian field) and the scaling limit, including a few comments about generalizations.
Finally, Section~\ref{sec-proofs} contains all the proofs, in particular asymptotics for the eigenvalues of the discrete fractional Laplacian.
\section{Notation and definitions}
\label{sec-not-and-def}
In this section, we will introduce all necessary notations and definitions.
Let $\mathbb{T}^d$ denote the d-dimensional torus, also defined as $\left [-\frac{1}{2},
\frac{1}{2}\right )^d \subset \mathbb{R}^d$. We will denote by $o$ the origin. The discretization will be denoted by
$\mathbb{T}^d_n := \left [-\frac{1}{2}, \frac{1}{2}\right )^d \cap (n^{-1} \mathbb{Z})^d$ for all $n \in \mathbb{N}$
and finally the discrete torus of side-length $n$ is denoted by
$\mathbb{Z}^d_n:=\left [-\frac{n}{2}, \frac{n}{2} \right )^d \cap \mathbb{Z}^d$.
We call $B(x,r)$ the ball centered at $x$ with radius $r$ in the $l^{\infty}(\mathbb{R}^d)$-metric and  $B_2(x,r)$ the corresponding ball in the  Euclidean metric (which will be denoted by $\|\cdot\|$).
In order to shorten the already lengthy notation, we will also use
$\|\cdot\|$  to denote the $L^2(\mathbb{R}^d)$ norm,  $\|\cdot\|_p$, to denote the $L^p(\mathbb{R}^d)$ norm. Moreover, we will simply denote by $\| \cdot \|_{D}$ the $L^{\infty}(D)$-norm in some domain $D\subset \mathbb{R}^d$.
Constants simply named $c$ or $C$ will always be positive, depending at most on $\alpha$ and $d$.
However, their values might change from line to, but their exact values are not relevant to our purposes.

\subsection{Long-range divisible sandpile models and discrete  fractional Laplacians}
\label{subsec-long-range-model}
First, we will define long-range random walks $(X_t)_{t\in \mathbb{N}}$ on the torus $\mathbb{Z}^d_n$. Let $\alpha \in (0,\infty)$ and consider the transition probabilities
$p^{(\alpha)}_n: \mathbb{Z}^d_n \times \mathbb{Z}^d_n \longrightarrow [0,1]$ defined by
$p^{(\alpha)}_n(x,y):=p^{(\alpha)}_n(o,y-x)$, with
\begin{equation}\label{def-random-walk-in-zdn}
	p^{(\alpha)}_n(o,x): =
	c^{(\alpha)}
	\sum_{\substack{z \in  \mathbb{Z}^d\backslash \{o\} \\ z \equiv x
	\!\!\!\!\!
	\mod \mathbb{Z}^d_n }} \frac{1}{\|z\|^{d+\alpha}} ,
\end{equation}
where $c^{(\alpha)}= \left(\sum_{z \in  \mathbb{Z}^d \backslash \{o\}} \frac{1}{\|z\|^{d+\alpha}} \right)^{-1}$ is the constant, such that $\sum_{x \in \mathbb{Z}^d_n} p^{(\alpha)}_n(o,x) = 1$. $x \equiv z \mod \mathbb{Z}^d_n$ will be short for $x_j \equiv z_j \mod n$ for all $j \in \{1,\dots,d\}$. Write from now on $p^{(\alpha)}_n(x):=p^{(\alpha)}_n(o,x)$.

Let $\mathbf{P}_x$  be
the law of the random walk $(X_t)_{t \ge 0}$ on $\mathbb{Z}^d_n$ starting at $x$,
with transition probabilities given by \eqref{def-random-walk-in-zdn} resp.~$\mathbf{E}_x$ its expectation.
We also define $\tau_z:=\inf\{t\ge 0: X_t =z\}$ and
$g^{(\alpha)}_z(x,y):=\mathbf{E}_x \left[\sum_{t=0}^{\tau_z-1} \1_{\{X_t=y\}} \right]$.
Call furthermore
\begin{equation}\label{char-field-long-range-divisible-eq-3}
	g^{(\alpha)}(x,y)
:=
	\frac{1}{n^d} \sum_{z \in \mathbb{Z}^d_n }g^{(\alpha)}_z(x,y).
\end{equation}
The \textit{discrete fractional Laplacian} on $\mathbb{Z}^d_n$ is the generator of $(X_t)_{t\ge0}$ and  is given by
\begin{align}\label{def-alpha-lap}
\nonumber
	-(-\Delta)^{\sfrac{\alpha}{2}}_{n} f (x)
&:=
	\Big(\sum_{y \in \mathbb{Z}^d_n} f(y)p^{(\alpha)}_n(x-y) \Big)- f(x)
\\&=
	\frac{1}{2}\sum_{y \in \mathbb{Z}^d_n} (f(x+y)+f(x-y) - 2f(x))p^{(\alpha)}_n(y),
\end{align}
where $f:\mathbb{Z}^d_n \longrightarrow \mathbb{R}$.

For $\alpha \in (0,2)$, we can define the \textit{continuous fractional Laplacian} of a periodic function $f \in C^\infty(\mathbb{T}^d)$ as
\begin{equation}\label{def-fractional-laplacian-1}
	-(-\Delta)^{\sfrac{\alpha}{2}}f(x)
:=
	\frac {2^{\alpha} \Gamma (\frac{d+\alpha}{2})}{\pi^{d/2}|\Gamma(-\frac{\alpha}{2})|}
	\int_{\mathbb{R}^d}\frac{f(x+y)+f(x-y)-2f(x)}{\|y\|^{d+\alpha}} \text{d} y.
\end{equation}
More precisely, $f \in C^{\infty}(\mathbb{R}^d)$ and $f(\cdot + e_i)=f(\cdot)$ for all $i =1, \dots,d$ where $\{e_i\}_{i=1}^d$ is the canonical basis of $\mathbb{R}^d$ and the integral above is defined in the sense of principal value. The constant in front of the integral is chosen to guarantee that for  $\alpha,\beta \in (0,2)$ such that $\alpha +\beta <2$, we have $(-\Delta)^{\sfrac{\alpha}{2}}(-\Delta)^{\sfrac{\beta}{2}}f = (-\Delta)^{\sfrac{(\alpha+\beta)}{2}}f$ for all $f \in C^\infty(\mathbb{T}^d)$. In Subsection~\ref{subsec-aws} we will introduce an equivalent definition of the fractional Laplacian in \eqref{def-fractional-laplacian-2}.  Alternative definitions of fractional Laplacians can be found in \cite{Kwasnicki15}.

Let us remark that  for all $\alpha \in (0,2)$, $f \in C^\infty(\mathbb{T}^d)$ and all $x \in \mathbb{T}^d$ we have
\begin{align}\label{eq-conv-of-fractional-laplacians}
	\lim_{n \longrightarrow \infty}n^\alpha  (-\Delta)^{\sfrac{\alpha}{2}}_{n }f(nx)
&=
	\frac{c^{(\alpha)}\pi^{d/2}|\Gamma(-\frac{\alpha}{2})|}{ 2^{\alpha} \Gamma (\frac{d+\alpha}{2})} (-\Delta)^{\sfrac{\alpha}{2}}f(x).
\end{align}

\begin{definition}
A \textit{divisible sandpile configuration} $s$ is a function $s:\mathbb{Z}^d_n \rightarrow \mathbb{R}$.
\end{definition}

For $x \in \mathbb{Z}^d_n$, if $s(x)\ge 0$, one should think of $s(x)$ as the quantity of mass at the site $x$. If $s(x)<0$, it can be interpreted as the size of a hole in $x$. If a site $x$ has mass $s(x) >1$, we call it \textit{unstable} and otherwise \textit{stable}. We then evolve the sandpile according to the following dynamics: unstable vertices will \text{topple} by keeping mass $1$ and distributing the excess over the other vertices proportionally according to the transition probabilities $p^{(\alpha)}_n$ at each discrete time step. Note that unstable sites in long-range divisible sandpile models distribute mass to \textit{all} vertices (including itself) at every time step, contrary to nearest neighbour divisible sandpile models which distribute mass only to their nearest neighbours. One could generate a divisible sandpile on  a graph from any random walk defined on it, where at each time step the mass is distributed proportional to the transition probabilities. We will elaborate more on this possibility in Remark \ref{rem:genRW}.

Let $s_t=(s_t(x))_{x\in \mathbb{Z}^d_n}$ denote the sandpile configuration after
$t\in \mathbb{N}$ discrete time steps (set $s_0:=s$ the initial configuration).
The parallel toppling procedure is given by the following algorithm.

\begin{algorithm}[Long-range divisible sandpile]\label{alg-divisible-long-range}
Set $t=1$, then run the following loop:
\begin{enumerate}
	\item if $\max_{x \in \mathbb{Z}^d_n} s_{t-1}(x) \le 1$, stop the algorithm;
	\item for all $x \in \mathbb{Z}^d_n$, set $e_{t-1}(x):=(s_{t-1}(x)-1)^+$;
	\item set $s_t(x):=s_{t-1}(x) - (-\Delta)^{\sfrac{\alpha}{2}}_n e_{t-1}(x) $;
	\item increase the value of $t$ by $1$ and go back to step $1$.
\end{enumerate}
\end{algorithm}

\begin{definition}
For each $t > 0$, the odometer function is a map $u^{\alpha}_t: \mathbb{Z}^d_n \rightarrow [0,\infty)$, defined as
\[
  u^\alpha_t(x):= \sum_{i=0}^{t-1} e_{i}(x)
\]
for all $x\in \mathbb{Z}^d_n$.
Using the fact that for each $x \in \mathbb Z^d_n$, $u^\alpha_t(x)$ is non-decreasing in $t$,
the limit $\lim_{t\rightarrow \infty} u^{\alpha}_t(x)$ is well-defined in $\mathbb R\cup \{\infty\}$,
for all $x\in \mathbb{Z}^d_n$. We will denote such a limit by $u^\alpha_\infty(x)$.
\end{definition}
Analogously to Section~2 in \cite{Levine2016} we have for every $x\in \mathbb{Z}^d_n$ and $t>0$:
\begin{equation}\label{equilibrium-equation-long-range-divisible-sandpile}
	s_{t}(x)
=
	s_0(x) - (-\Delta)^{\sfrac{\alpha}{2}}_{n} u^{\alpha}_{t}(x).
\end{equation}
From \cite{Levine2016} we have the following dichotomy: either for all $x\in \mathbb{Z}^d_n$ we have \textit{stabilisation}, 
i.e. $u^{\alpha}_{\infty}(x)<\infty$ or \textit{explosion}, i.e.  for all $x\in \mathbb{Z}^d_n:$ $u^{\alpha}_{\infty}(x)=\infty$.
We will see in Lemma~\ref{prop:stabilisation}
that given an initial configuration $s_0$, satisfying $\sum_{x \in \mathbb{Z}^d_n}s_0(x)= n^d$,  we  have $u^{\alpha}_{\infty}(x)<\infty$ for all $x\in \mathbb{Z}^d_n$ and $s_{\infty}\equiv 1$.

It is important to notice that the long-range divisible sandpile can be studied in 
terms of other toppling procedures as well, see Definition~2.1 in \cite{Levine2016}. Moreover, the abelian property and least
action principle, see Proposition~2.5 in \cite{Levine2016}, can be proved using 
essentially the same techniques used for the nearest neighbour divisible sandpile. 

Define the initial configuration $s_0$ for $x\in \mathbb{Z}^d_n$ by
\[
	s_0( x ) :=
	1+ \sigma (x) - \frac {1}{n^d} \sum_{y \in \mathbb{Z}^d_n} \sigma(y),
\]
where $(\sigma(x))_{x \in \mathbb{Z}^d_n}$ is a collection of i.i.d random variables with $ \mathbb{E}[\sigma(x)]=0$ and $\var [\sigma(x)]=1$. $s_0$ chosen in this way guarantees that $\sum_{x\in \mathbb{Z}^d_n} s_0(x)=n^d$. 
We will show in Proposition~\ref{prop:rescaledGauss} the following equality in law
\begin{equation}\label{char-field-long-range-divisible-eq-1}
	(u^{\alpha}_\infty(x) )_{x \in \mathbb{Z}^d_n}
\overset{d}=
	\Big(
		\eta^{\alpha}(x)- \min_{z \in \mathbb{Z}^d_n}(\eta^{\alpha}(z))
	\Big)_{x \in \mathbb{Z}^d_n },
\end{equation}
where $(\eta^{\alpha}(x))_{x\in \mathbb{Z}^d_n}$ are defined by
\[
	\eta^\alpha(x) := \sum_{y \in \mathbb{Z}^d_n}g^{(\alpha)}(x,y)(s_0(y)-1),
\]
and $g^{(\alpha)}$ was defined in \eqref{char-field-long-range-divisible-eq-3}.

Note that the distribution of $\eta^\alpha$ is invariant by translations.
Moreover, it has mean $0$ and covariance given by
\begin{equation}\label{char-field-long-range-divisible-eq-2}
	\mathbb{E}[\eta^{\alpha}(x) \eta^{\alpha}(y)] =
	\sum_{w \in \mathbb{Z}^d_n} g^{(\alpha)}(x,w)g^{(\alpha)}(w,y).
\end{equation}
We can see easily that the covariance solves the equation
\[
	\Big(-(-\Delta)^{\sfrac{\alpha}{2}}_{n }\Big)^2
	\Big[\mathbb{E}[\eta^{\alpha}(x) \eta^{\alpha}(y)]\Big]
	=
	\delta_x(y) - \frac{1}{n^d}.
\]
Remark that when $(\sigma(x))_{x \in \mathbb{Z}^d_n}$ are i.i.d. Gaussians,  $(\eta^{\alpha}(x))_{x \in \mathbb{Z}^d_n}$ can be interpreted as a massive discrete fractional Gaussian field on $\mathbb{Z}^d_n$.

\begin{figure}[ht]
  \begin{subfigure}[b]{0.5\linewidth}
    \centering
    \includegraphics[width=0.8\linewidth]{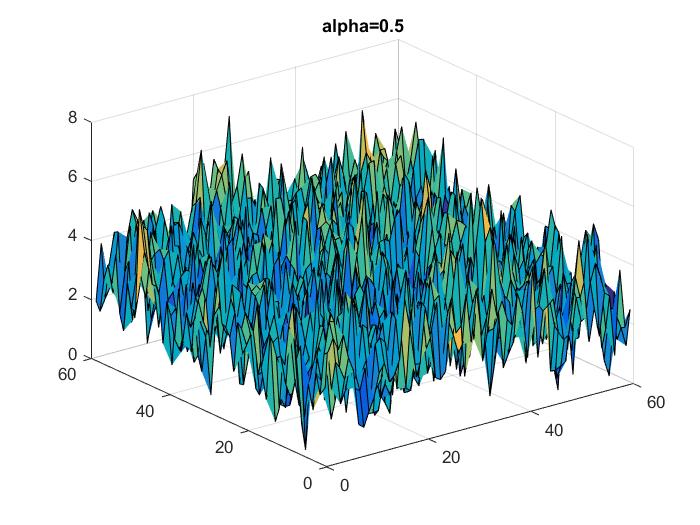}
    \vspace{4ex}
  \end{subfigure}
  \begin{subfigure}[b]{0.5\linewidth}
    \centering
    \includegraphics[width=0.8\linewidth]{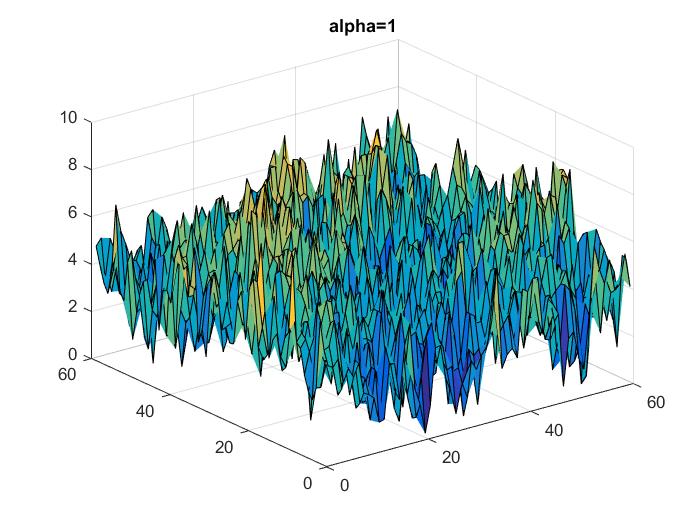}
    \vspace{4ex}
  \end{subfigure}
  \begin{subfigure}[b]{0.5\linewidth}
    \centering
    \includegraphics[width=0.8\linewidth]{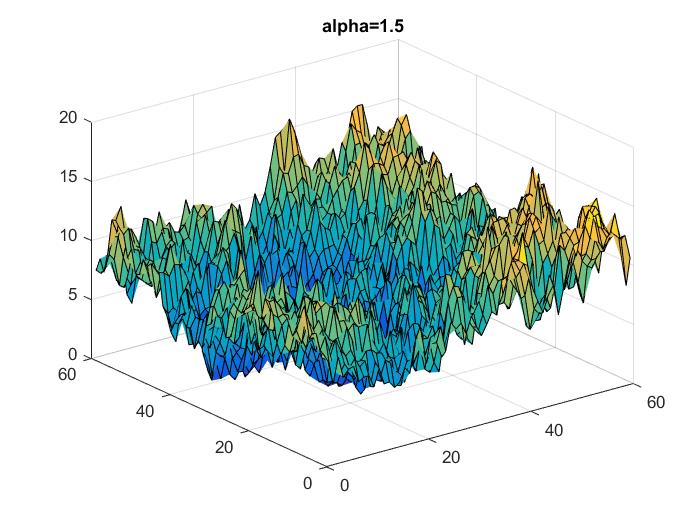}
  \end{subfigure}
  \begin{subfigure}[b]{0.5\linewidth}
    \centering
    \includegraphics[width=0.8\linewidth]{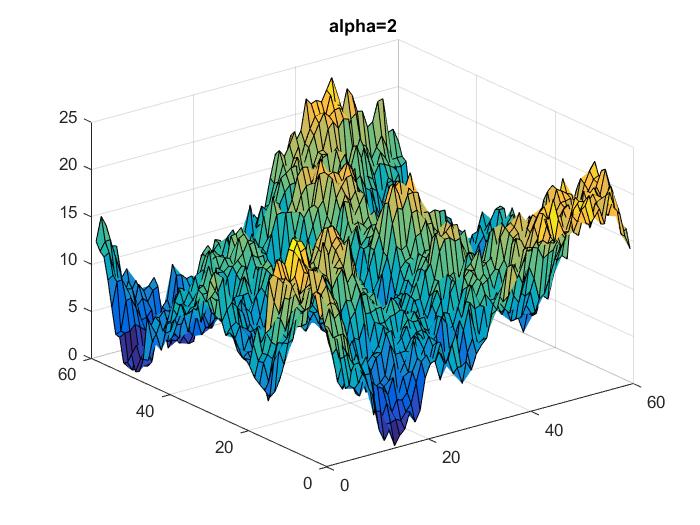}
  \end{subfigure}
  \caption{ Simulations of the odometer for different values of $\alpha\in [0.5,2] $ in the discrete torus of length $60$
	and standard Gaussian initial random variables.}
  \label{Fig1}
\end{figure}

\begin{remark}
	One might feel tempted to write $\Big((-\Delta)^{\sfrac{\alpha}{2}}_{n }\Big)^2=(-\Delta)^{\alpha}_{n}$, however this is not correct in the discrete case. Such property is valid in the continuous case because fractional Laplacians are fractional powers of each other. It fails in the discrete case as $\mathbb{Z}^d$ is not invariant by arbitrary rotations. The easiest way of seeing that, is to study the eigenvalues of $(-\Delta)_n^{\sfrac{\alpha}{2}}$. In case the property was valid, there should be a constant $c = c(\alpha,d,n)$ such that $(\lambda^{(\alpha,n)}_w)^2 = c \lambda^{(2\alpha,n)}_w$ which is not true. For more discussion on the fractional powers of the discrete Laplacian, we refer to \cite{Ciaurri2015}. However, for $\alpha, \beta \in (0,2)$ such that $\alpha+\beta < 2$, we have
\[
	n^{\alpha+\beta} c^{-1}_{d,\alpha}c^{-1}_{d,\beta}
	(-\Delta)^{\sfrac{\beta}{2}}_{n }(-\Delta)^{\sfrac{\alpha}{2}}_{n }f(nx)
	\longrightarrow (-\Delta)^{\sfrac{(\alpha+\beta)}{2}}f(x),
\]
as $n \rightarrow \infty$, where $c_{d,\alpha}$ and $c_{d,\beta}$ are
the respective constants found on the right-hand side of \eqref{eq-conv-of-fractional-laplacians}.
Therefore, the powers of the fractional Laplacians are additive in the limit.
\end{remark}

\subsection{Fourier analysis on the torus}
\label{subsec-fourier-torus}

We will use the following inner product for the space $\ell^2(\mathbb{Z}^d_n)$, the Hilbert space of complex valued functions on the discrete torus,

\[
	\langle {f,g} \rangle
	=
	\frac{1}{n^d} \sum_{z \in \mathbb{Z}^d_n}
	f(z) \overline{g(z)}.
\]
Consider then the Fourier basis  given by the eigenfunctions $(\psi_w)_{w\in \mathbb{Z}^d_n}$ with
\begin{equation}\label{def-fourier-basis-discrete}
	\psi_w(z)=\psi^{(n)}_w(z):=\exp \Big(2\pi i z \cdot\frac{w}{n}\Big).
\end{equation}
Given $f \in \ell^2(\mathbb{Z}^d_n)$, we define the discrete Fourier transform by
\[
	\widehat{f}(w)
=
	\langle  f, \psi_{w}\rangle
=
	\frac{1}{n^d}\sum_{z \in \mathbb{Z}^d_n}
	f(z) \exp \Big(-2\pi i z\cdot\frac{w}{n}\Big).
\]
Similarly, if $f, g \in L^2(\mathbb{T}^d)$ we will denote the inner product by
\[
	(f,g)_{L^2(\mathbb{T}^d)}:= \int_{\mathbb{T}^d} f(z) \overline{g(z)}
	\text{d}z.
\]
Consider the Fourier basis of $L^2(\mathbb{T}^d)$, given by the eigenvectors $\phi_\nu(x):=\exp(2\pi i \nu \cdot x)$, $\nu \in \mathbb{Z}^d$, and  denote the Fourier transform by
\[
	\widehat{f}(\nu):=(f,\phi_\nu)_{L^2(\mathbb{T}^d)} =
	\int_{\mathbb{T}^d} f(z) e^{ - 2 \pi i \nu \cdot z} \text{d}z.
\]
In this article, we will use $\widehat{\cdot}$ to refer to both the Fourier transform in $\ell^2(\mathbb{Z}_n^d)$ and in $L^2(\mathbb{T}^d)$, which will be clear from the context. However, it will be important to notice that for $f \in C^\infty(\mathbb{T}^d)$, if we define $f_n: \mathbb{Z}^d_n \longrightarrow \mathbb{R}$ by $f_n(z) = f( \frac{z}{n})$, then for all $w \in \mathbb{Z}^d$, $\widehat{f_n}(w) \longrightarrow \widehat{f}(w)$
as $n\rightarrow \infty$.
Let us show that $(\psi_w)_{w \in \mathbb{Z}^d_n}$ are indeed eigenvectors of $-(-\Delta)^{\sfrac{\alpha}{2}}_n $. Let $w\in \mathbb{Z}^d_n \backslash \{o\}$, then
\begin{align*}
	-(-\Delta)^{\sfrac{\alpha}{2}}_n \psi_w(x)
	&=
	\sum_{y \in \mathbb Z^d_n} p^{(\alpha)}_n(y-x)(\psi_w(y) - \psi_w(x))
	\\ &=
	\psi_w(x)\sum_{y \in \mathbb Z^d_n} p^{(\alpha)}_n(y-x)(\psi_w(y-x) - \psi_w(o))
	\\ &=
	\psi_w(x)\sum_{z \in \mathbb Z^d_n} p^{(\alpha)}_n(z)(\psi_w(z) - \psi_w(o))
\end{align*}
and the last sum does not depend on $x$. We will denote by $\lambda^{(\alpha,n)}_w$ the respective eigenvalue, and we will properly
evaluate it in Section \ref{sec-proofs}. The following will be a very useful identity
for the discrete Fourier transform $ \widehat{g^{(\alpha)}}(x,w)$ of $g^{(\alpha)}(x,\cdot)$ when $w \neq o$,
\begin{align}\label{identity-fourier-of-green}
	\lambda^{(\alpha,n)}_w	 \widehat{g^{(\alpha)}}(x,w)
&=
	\lambda^{(\alpha,n)}_w
	\langle g^{(\alpha)}(x,\cdot), \psi_w (\cdot) \rangle
\\& =
	\langle g^{(\alpha)}(x,\cdot),
	-(-\Delta)^{\sfrac{\alpha}{2}}_{n} \psi_w (\cdot) \rangle
\nonumber  =
	\langle -(-\Delta)^{\sfrac{\alpha}{2}}_{n} g^{(\alpha)}(x,\cdot),
	\psi_w (\cdot) \rangle
\nonumber\\ & =
	-\langle \delta_x,
	\psi_w (\cdot) \rangle =
	-\frac{1}{n^d} \psi_{-w}(x).
\end{align}

\subsection{Abstract Wiener Spaces and continuum fractional Laplacians}
\label{subsec-aws}
We need to define an abstract Wiener space (AWS) appropriately since the scaling limit will be a random distribution. Let us remark that we have to construct a different AWS than in \cite{Cipriani2016}, since we are dealing with general fractional Gaussian fields.  Our presentation is based on Section~2 in \cite{Sheffield07} and Sections~6.1, 6.2 in \cite{Silvestri17}.

An {\em abstract Wiener space} (AWS) is a triple $(H, B, \mu)$, where:
\begin{enumerate}
    \item $(H, (\cdot,\cdot)_H)$ is a Hilbert space;
    \item $(B, \|\cdot\|_{B})$ is the Banach space completion of $H$ with respect to the measurable norm $\|\cdot\|_B$ on $H$, equipped with the Borel $\sigma$-algebra $\mathcal B$ induced by $\|\cdot\|_B$; and
    \item $\mu$  is the unique Borel probability measure on $(B,\mathcal B)$ such that, if $B^*$ denotes the dual space of $B$, then $\mu\circ\phi^{-1}\sim \mathcal N(0, \|\widetilde \phi\|^2_{H}) $ for all $\phi\in B^*$, where $\widetilde \phi$ is the unique element of $H$ such that $\phi(h)=(\widetilde \phi, h)_H$ for all $h\in H$.
\end{enumerate}
Note that, in order to construct a measurable norm $\|\cdot\|_B$ on $H$, it suffices to find a Hilbert-Schmidt operator $T$ on $H$, and set $\|\cdot\|_B :=\|T\cdot\|_H$.

Let us present the class of AWS which we will study and which is connected to the fractional powers of the Laplacian. Consider again $(\phi_\nu)_{\nu\in \mathbb{Z}^d}$ as the Fourier basis of $L^2(\mathbb{T}^d)$ given in the previous subsection, we have $(\phi_\nu)_{\nu\in \mathbb{Z}^d}$ is a basis of eigenvectors of $-(-\Delta)^{\sfrac{\alpha}{2}}$, satisfying
\[
    (-\Delta)^{\sfrac{\alpha}{2}}\phi_\nu = \|\nu\|^\alpha \phi_\nu.
\]
Also notice that
\[
    (-\Delta)\phi_\nu = \|\nu\|^2 \phi_\nu,
\]
for the usual Laplacian. Hence,  we can extend the definition \eqref{def-alpha-lap} of the discrete fractional Laplacian to $L^2(\mathbb{T}^d)$-functions in a very natural way, which also supports any power $a \in \mathbb{R}$ of $(-\Delta)$. Let $f \in L^2(\mathbb{T}^d)$ with Fourier expansion $\sum_{\nu\in \mathbb{Z}^d} \widehat{f}(\nu) \phi_\nu(\cdot)$,
and $a\in \mathbb{R}$. We define the operator $(-\Delta)^a$ as
\begin{equation}\label{def-fractional-laplacian-2}
    (-\Delta)^a f (\cdot)
=
    \sum_{\nu\in \mathbb{Z}^d \backslash\{o\}}\|\nu\|^{2a}\widehat f(\nu)\mathbf \phi_{\nu}(\cdot).
\end{equation}
For all $a \in \mathbb{R}$, $(-\Delta)^a(f)=0$ for all constant functions, hence we can study the operator  $(-\Delta)^a$ acting only on functions $f \in L^2(\mathbb{T}^d)$ such that $\int_{\mathbb{T}^d} f(z) \text{d}z=0$. With this in mind, let ``$\sim$'' be the equivalence relation on $C^\infty(\mathbb{T}^d)$ which identifies two functions differing by a constant. Let $H^a=H^a(\mathbb{T}^d)$ be the Hilbert space completion of $C^\infty(\mathbb{T}^d)/{\sim}$ under the norm
\begin{equation}\label{eq:norm}
	(f,\,g)_{a}
:=
	\sum_{\nu\in \mathbb{Z}^d\backslash \{o\}}
	\|\nu\|^{4a}\widehat {f}(\nu) \overline{\widehat{g}(\nu)}.
\end{equation}
Define the Hilbert space
\[
    \mathcal H_{a}:=\Big\{u\in L^2(\mathbb{T}^d):\,(-\Delta)^a u\in L^2(\mathbb{T}^d)\Big\}/{\sim}
\]
equipped with the norm
\begin{equation}\label{def-sob-norm}
    \|f\|_{\mathcal H_{a}(\mathbb{T}^d)}^2
=
    \Big((-\Delta)^a f,\,(-\Delta)^a f\Big)_{L^2(\mathbb{T}^d)}.
\end{equation}
In fact, $(-\Delta)^{-a}$ provides a Hilbert space isomorphism between $\mathcal{H}_a$ and $H^a$, which we identify when needed.
For
\begin{equation}\label{ineq-epsilon}
    b<a-\frac{d}{4}
\end{equation}
one shows that $(-\Delta)^{b-a}$ is a Hilbert-Schmidt operator on $H^a$ (cf.~also \cite[Proposition~5]{Silvestri17}). In our case, we will be setting $a:=-\frac{\gamma}{2}$, where $\gamma := \min\{\alpha,2\}$. Therefore, by \eqref{ineq-epsilon},
for any $-\varepsilon:=b<0$ which satisfies $\varepsilon>\frac{\gamma}{2}+\frac{d}{4}$, we have that $(H^{-\frac{\gamma}{2}},\,\mathcal{H}_{-\varepsilon},\,\mu_{-\varepsilon})$ is an AWS. The measure $\mu_{-\varepsilon}$ is the unique Gaussian law on $\mathcal H_{-\varepsilon}$ whose characteristic functional is equal to
\[
    \Phi(f):=\exp\left(-\frac{\|f\|_{{-\frac{\gamma}{2}}}^2}{2}\right).
\]
The norm $\| \cdot \|_{-\gamma/2}$ is defined in \eqref{eq:norm} taking $a=-\gamma/2$.
The field associated to $\Phi$ is called (continuous) fractional
Gaussian Field with paramater $\gamma$, and it will be denoted by either
$\Xi^\gamma$ or fGF$_{\gamma}(\mathbb{T}^d)$. It corresponds to the limiting field
appearing in Theorem~\ref{theorem-main-non-Gaussian}.

\section{Results}\label{sec-res}
\subsection{Stabilisation and law of the odometer on $\mathbb{Z}^d_n$}
\label{subsec-stabilisation}

The following lemma is a simple result concerning stabilisation of a divisible sandpile model. The proof is analogous to the counterpart in the nearest neighbours case, which can be found in Lemma~7.1 in \cite{Levine2016} and will be left for the reader. We consider $\alpha \in (0,\infty)$ and the toppling defined according to Algorithm~\ref{alg-divisible-long-range}.

\begin{lemma}\label{prop:stabilisation}
Let $s_0: \mathbb{Z}^d_n \longrightarrow \mathbb{R}$ be any initial configuration
    satisfying $\sum_{x \in \mathbb{Z}^d_n} s_0(x) = n^d$. Then $s$ stabilises to
    the all $1$ configuration and its odometer $u^{\alpha}_{\infty}$ is the unique function satisfying
    $s_0(x)-(-\Delta)^{\sfrac{\alpha}{2}}_{n} u_{\infty}^{\alpha}(x) =1$ for all $x\in \mathbb{Z}^d_n$ and
    $\min_{x \in \mathbb{Z}^d_n } u_{\infty}^{\alpha}(x)=0$.
\end{lemma}

Applying the above result, in an analogous manner as in Proposition~1.3 in \cite{Levine2016}, we get the following result.

\begin{proposition}\label{prop:rescaledGauss}
Let $(\sigma(x))_{x \in \mathbb{Z}^d_n}$ be i.i.d such that $\mathbb{E}[\sigma(x)]=0$ and $\var[\sigma(x)]=1$. Consider the long-range divisible sandpile with initial condition
    \[
        s_0(x)=
        1+ \sigma(x) - \frac{1}{n^d}\sum_{y \in \mathbb{Z}^d_n}  \sigma(y).
    \]
Then $s$ stabilises to the all $1$ configuration and the distribution of the  odometer $u_{\infty}^{\alpha} : \mathbb{Z}^d_n \rightarrow \mathbb{R}$ is equal  to 
\[
    (u_{\infty}^{\alpha}(x))_{x \in \mathbb{Z}^d_n } \overset{d}= \left(\eta^{\alpha}(x) - \min_{z\in \mathbb{Z}^d_n} (\eta^{\alpha}(z)) \right)_{x\in \mathbb{Z}^d_n}.
\]
 $\eta^{\alpha}$ is given by
\begin{align*}
    \eta^{\alpha}(x)
&=
    \sum_{z \in \mathbb{Z}^d_n}g^{(\alpha)}(x,z) (s_0(z)-1)
\\&=
    \sum_{z \in \mathbb{Z}^d_n}g^{(\alpha)}(x,z)
    \Big(\sigma(z)- \frac{1}{n^d} \sum_{y \in \mathbb{Z}^d_n}\sigma(y)\Big)
\end{align*}
with $g^{(\alpha)}$ defined as in \eqref{char-field-long-range-divisible-eq-3} and $x\in \mathbb{Z}^d_n$.
In particular, 
\[
    \mathbb{E}[\eta^{\alpha}(x)\eta^{\alpha}(y)]
=
    \sum_{z \in \mathbb{Z}^d_n}g^{(\alpha)}(x,z) g^{(\alpha)}(z,y).
\]
\end{proposition}

\subsection{The expected odometer on the finite torus}
\label{subsec-mean-value-odometer}

In this section we ask how the behaviour of the odometer is affected by the introduction of the long-range distribution on the finite grid $\mathbb{Z}^d_n$ when $(\sigma(x))_{x\in \mathbb{Z}^d_n}$ are i.i.d.~ standard Gaussians. We will prove here the equivalent version of Theorem~1.2 from \cite{Levine2016}.

\begin{theorem}\label{thm-bounds-odometer-finite}
	Let $\alpha \in \mathbb{R}_+\backslash \{2\}$, $d \ge 1$ and $(\sigma(x))_{x \in \mathbb{Z}^d_n}$  i.i.d standard normal random variables. Furthermore, let $s_0$ be the initial sandpile configuration given by
    \[
        s_0(x)= 1 + \sigma(x) - \frac{1}{n^d}\sum_{y \in \mathbb{Z}^d_n} \sigma(y), \, \, x\in \mathbb{Z}^d_n
    \]
and the redistribution rule defined by Algorithm \ref{alg-divisible-long-range}.
    Then $s$ stabilises to the all $1$ configuration and there exists a positive constant $C_{d,\alpha}>0$, such that the final odometer $u^{\alpha}_{\infty}$ satisfies for all $x\in \mathbb{Z}^d_n$
    \[
        C^{-1}_{d,\alpha} \Phi_{d,\gamma}(n)
\le
        \mathbb{E}[u^\alpha_\infty(x)]
\le
        C_{d,\alpha} \Phi_{d,\gamma}(n),
    \]
    where $\gamma:= \min\{\alpha,2\}$ and $\Phi_{d,\gamma}$ is given by
    \begin{equation}
        \Phi_{d,\gamma}(n)
:=
        \begin{cases}
            n^{\gamma- \frac{d}{2}}    ,& \text{ if } \gamma > \frac{d}{2}\\
            \log(n)                        ,& \text{ if } \gamma = \frac{d}{2}\\
           (\log(n))^{\frac{1}{2}}, & \text{ if } \gamma < \frac{d}{2}.
        \end{cases}
    \end{equation}
\end{theorem}

Let us make two remarks about this result. First, note that for $\alpha>2$, comparing the result above with its counterpart Theorem~1.2 in \cite{Levine2016}, the asymptotic behaviour of the expected odometer is
the same as for the nearest-neighbours divisible sandpile model. Secondly,
for $\alpha=2$ we expect that $\mathbb{E}[u^\alpha_\infty(x)]$ behaves like $\Phi_{d,2}(n)$ times some
$\log(n)$ factors that might depend on the dimension.

\subsection{Scaling limit of the odometer}
\label{subsec-scaling-odometer}

\begin{theorem} \label{theorem-main-non-Gaussian}
Let $\alpha \in \mathbb{R}_+$, $d \ge 1$, assume $(\sigma(x))_{x\in \mathbb{Z}^d_n}$ is a collection of i.i.d.~random variables with
 $\var [\sigma(x)]=1$ for all $x\in \mathbb{Z}^d_n$.
Consider the long-range divisible
sandpile in $\mathbb{Z}^d_n$ with initial configuration
\[
     s_0( x ) = 1+ \sigma ( x ) -
     \frac { 1} { n^d } \sum _ { y \in \mathbb{Z}^d_n } \sigma ( y )
\]
and redistribution defined by Algorithm \ref{alg-divisible-long-range}. 
Define the formal field on $\mathbb{T}^d$ by
\begin{equation}\label{def-aproaching-field-eq}
    \Xi^{\alpha}_n(x):=  \frac{a_\alpha(n)}{\tilde{c}^{(\alpha)}}
    \sum_{z \in  \mathbb{T}^d_n}
    u^{\alpha}_{\infty}(nz)\1_{B\left(z,\frac{1}{2n}\right)}(x), \, \, x \in \mathbb{T}^d,
\end{equation}
where
\[
  \tilde{c}^{(\alpha)}:=
  \begin{cases}
    \lim_{n\rightarrow \infty} \frac{-n^{\alpha} \lambda^{(\alpha,n)}_w}{\|w\|^{\alpha}} >0, & \text{ if } \alpha<2 \\
    \frac{c^{(2)}\pi^{\frac{d+4}{2}}}{d\cdot\Gamma(d/2)}, &
    \text{ if } \alpha=2 \\
     \sum_{x \in \mathbb Z^d \setminus \{o\}}  \frac{c^{(\alpha)}\pi^2 x_1^2}{\|x\|^{d+\alpha}}, &
    \text{ if } \alpha>2
  \end{cases}
\]
and 
\[
    a_\alpha(n)
=
\begin{cases}
    n^{\frac{d-2\alpha}{2}},  & \text{if } \alpha<2;\\
    n^{\frac{d-4}{2}}\sqrt{\log(n)}, & \text{if } \alpha=2;\\
    n^{\frac{d-4}{2}}, & \text{if } \alpha>2.
\end{cases}
\]

We identify
$\Xi^{\alpha}_n$ with the distribution
acting on mean zero test functions $f \in C^\infty(\mathbb{T}^d)$ by
$(\Xi^\alpha_n, f) :=(\Xi^\alpha_n, f)_{L^2(\mathbb{T}^d)}$.
Then, we have that $\Xi_n^\alpha$ converges in law to a fractional Gaussian field
with parameter $\gamma$, denoted by
$\Xi^\gamma$ or fGF$_{\gamma}(\mathbb{T}^d)$, with mean zero and covariance defined by
\begin{equation}\label{thm-main-cov-eq}
    \mathbb{E} \Big( ( \Xi^{\gamma},f), ( \Xi^{\gamma},g ) \Big)
    =
    \sum_{w \in \mathbb{Z}^d \backslash \{o\}}
    \|w\|^{-2\gamma} \widehat{f}(w) \overline{\widehat{g}(w)},
\end{equation}
where $\gamma:= \min\{\alpha,2\}$.
This convergence holds in
$\mathcal{H}_{-\varepsilon}$ for
$\varepsilon>\max\{\frac{\gamma}{2} +\frac{d}{4}, \frac{d}{2}  \}$.
\end{theorem}

Let us emphasize again two special cases included in the result above. $\gamma=1$ corresponds to the GFF and $\gamma=2$ to the bi-Laplacian model. Note further that it is enough to prove the theorem in the case $\mathbb{E}[\sigma(o)]=0$. For random variables with non-zero mean write
\[
     s_0 ( x ) = 1+ \sigma ( x ) -
     \frac { 1} { n^d } \sum _ { y \in \mathbb{Z}^d_n } \sigma ( y )
     = 1+ \left(\sigma ( x ) - \mathbb{E}[\sigma(o)] \right)
     -\frac { 1} { n^d } \sum _ { y \in \mathbb{Z}^d_n } \left(\sigma ( y ) - \mathbb{E}[\sigma(o)] \right)
\]
which falls into the previous case. Let us discuss some further generalizations in the sequel.

\begin{remark}\label{rem:genRW}
Note that the redistribution of the mass, specified in Algorithm \ref{alg-divisible-long-range}, depends on $(-\Delta)^{\alpha/2}_n$ which is defined w.r.t.~ the long-range random walk with transition probabilities $p^{(\alpha)}_n$ given in \eqref{def-random-walk-in-zdn}. The fact that one obtains  the fractional Gaussian fields with parameter $\gamma=\min\{\alpha,2\}$ as scaling limits of the odometer should not depend on the particular law $p^{(\alpha)}_n$ but rather on moments of $X_1$.  We expect the following generalization to hold. Let $(X_t)_{t\ge0}$ be a random walk with transition probabilities given by
$p(x,y)=p(\|x-y\|)$. Define its periodisation by
\[
    p_n(x)
=
    \sum_{\substack{ z \in \mathbb{Z}^d \backslash \{o\} \\ z \equiv x \mod \mathbb{Z}^d_n }}p(z).
\]
Suppose that $p(\cdot)$ is in the domain of attraction of the $\alpha$-stable distribution.

Consider the divisible sandpile model on $\mathbb{Z}^d_n$ where the mass is distributed according to $p_n$. Denote its final odometer 
by $u_{\infty}^{(p)}$ and the formal field on $\mathbb{T}^d$ by
\[
    \Xi^{(p)}_n(x):= n^{\frac{d-2\alpha}{2}}
	\sum_{z \in  \mathbb{T}^d_n}
	u^{(p)}_{\infty}(nz) \1_{B(z,\frac{1}{2n})}(x), \, \, x\in \mathbb{T}^d
\]
We believe that $\Xi^{(p)}$ converges in law o the  fGF$_{\alpha}(\mathbb{T}^d)$.
\end{remark}

\begin{remark}
	We showed that if the initial configuration $s_0$ for the long-range divisible sandpile model is chosen in such a way that  $\sum_{x \in \mathbb{Z}^d_n} s_0(x) = n^d$,
	then the odometer $u^{\alpha}_{\infty}$ is finite a.s. and $s_\infty \equiv 1$. Consider now $$ s^\prime ( x )
	= 1 + c_0+ \sigma ( x ) - \frac { 1} { n^d } \sum _ { y \in \mathbb{Z}^d_n }
	\sigma ( y )$$ 
    for some $c_0 \in \mathbb{R}$. If $c_0>0$, then clearly $u^{\alpha}_\infty \equiv \infty$ for every realization.
	However, if we define
	\[
		\tilde{u}^{\alpha}_t(x) := u^{\alpha}_t(x)-\frac{1}{n^d} \sum_{y\in \mathbb{Z}^d_n}u^{\alpha}_t(y),
	\]
	we still have $s_t(x) = s'(x) - (-\Delta)^{\sfrac{\alpha}{2}}_n \tilde{u}^{\alpha}_t(x)$
for all $x\in \mathbb{Z}^d_n$. 
	In this case we can prove that, for all $x
	\in \mathbb{Z}^d_n$, $s_t(x) \rightarrow 1+ c_0$ and
	$\tilde{u}^{\alpha}_\infty(x) \in [0,\infty)$. The scaling
	limit of the field $\tilde{u}^{\alpha}_{\infty}$ is the same as of $u^{\alpha}_{\infty}$. However, for $c_0<0$ it is less clear what happens since we do not know how the configurations $s'$ and $s'_{\infty}$ correlate.
\end{remark}

Finally, let us remark that the asymptotics of Green functions can be used to recover the kernel of the fractional Laplacian for $\alpha \in (0,2)$ for dimension 	$d > 2 \alpha$. The proof is analogous to the one of Theorem~3 in \cite{Cipriani2016}, hence we leave it to the reader.

\section{Proofs} \label{sec-proofs}

\subsection{Estimates for the eigenvalues of discrete fractional Laplacians}\label{subsec-estimates-eigenvalues}

The proofs of  Theorem~\ref{thm-bounds-odometer-finite} resp.~Theorem~\ref{theorem-main-non-Gaussian} follow similar ideas as the proofs of Theorem~1.2 in \cite{Levine2016} resp.~Theorem~2 in \cite{Cipriani2016}. The main difference is exchanging the normalised graph Laplacian by the discrete fractional Laplacian given in \eqref{def-alpha-lap}. More specifically, we need a very sharp control over the eigenvalues associated to the discrete fractional Laplacian.

Note that for the nearest neighbour divisible sandpile model, one studies the normalised graph Laplacian  $\Delta_n: \ell^2(\mathbb{Z}^d_n) \rightarrow \ell^2(\mathbb{Z}^d_n)$ given by
\[
    \Delta_n f(x)= \frac{1}{2d} \sum_{\substack{y \in \mathbb{Z}^d_n \\ x \sim y}} (f(y) - f(x)),
\]
where $x \sim y$ denotes  nearest neighbours modulo $\mathbb{Z}^d_n$. Remark that in \cite{Cipriani2016} the authors consider the non-normalized Laplacian, but the factor $1/2d$ appears later in the definition of the discrete odometer $e_n$ in Proposition~4.
It is easy to see that, $(\psi_{w})_{w \in \mathbb{Z}^d_n}$ as described in Subsection~\ref{subsec-fourier-torus}, are eigenvectors of $\Delta_n$ with respective eigenvalues given by
\begin{equation}\label{discrete-eigenvalues-usual-laplacian-eq}
	\lambda^{(n)}_w = - \frac{2}{d} \sum_{i=1}^d \sin^2 \Big(\pi \frac{w_i}{n} \Big),
\end{equation}
which, once properly rescaled, are close to $\|\pi w \|^2$. However, the discrete fractional Laplacian $-(-\Delta)^{\sfrac{\alpha}{2}}_n$ has eigenvalues $(\lambda^{(\alpha,n)}_w)_{w\in \mathbb{Z}^d_n}$  associated with the eigenfunctions $(\psi_w)_{w\in \mathbb{Z}^d_n}$ where
\begin{align}\label{discrete-eigenvalues-eq}
    \lambda^{(\alpha,n)}_w
&=
    -c^{(\alpha)}\sum_{x \in \mathbb{Z}^d \backslash\{o\}}
    \frac{\sin^2(\pi x \cdot \frac{w}{n} )}{\|x\|^{d+\alpha}}
=    -c^{(\alpha)} \frac{1}{n^{d+\alpha}}
    \sum_{x \in \frac{1}{n} \mathbb{Z}^d \backslash\{o\}}
    \frac{\sin^2(\pi x \cdot w )}{\|x\|^{d+\alpha}},
\end{align}
and $c^{(\alpha)}$ is just the normalising constant of the associated long range-random walk in $\mathbb{Z}^d$.
This is a simple computation that uses the fact that $\psi_k(x+y)-\psi_k(x)=\psi_k(x)\left( \psi_k(y)-1 \right)$.
A quick comparison between the eigenvalues \eqref{discrete-eigenvalues-usual-laplacian-eq} and \eqref{discrete-eigenvalues-eq} shows that we will need to proceed with some extra care to understand the asymptotic behaviour of $\lambda^{(\alpha,n)}_w$ in terms of $n$ and $w$.

In fact, for $\alpha \in (0,2)$ one can easily show that
\[
    n^\alpha \lambda^{(\alpha,n)}_w
=
    -c^{(\alpha)} \frac{1}{n^{d}}
    \sum_{x \in \frac{1}{n} \mathbb{Z}^d \backslash\{o\}}
    \frac{\sin^2(\pi x \cdot w )}{\|x\|^{d+\alpha}} \stackrel{n \to \infty}\longrightarrow
     -c^{(\alpha)}\int_{\mathbb{R}^d}\frac{\sin^2(\pi z\cdot w)}{\|z\|^{d+\alpha}}\text{d}z.
\]
Let 
\begin{equation}\label{eq:ctilde}
\tilde{c}^{(\alpha)}:= c^{(\alpha)}\int_{\mathbb{R}^d}
 \frac{\sin^2(\pi z_1)}{\|z\|^{d+\alpha}} \text{d}z, 
\end{equation}
which stems from $\tilde{c}^{(\alpha)} = \lim_{n\rightarrow \infty} \frac{n^{\alpha} \lambda_w^{(\alpha, n)}}{\|w\|^{\alpha}}$,
then
\begin{align}\label{continuous-eigenvalue}
\nonumber
    \lambda^{(\alpha,\infty)}_w
&:=
     -c^{(\alpha)}\int_{\mathbb{R}^d}\frac{\sin^2(\pi z\cdot w)}{\|z\|^{d+\alpha}}\text{d}z
\nonumber\\&=
     -c^{(\alpha)} \|w \|^\alpha \int_{\mathbb{R}^d}
     \frac{\sin^2\Big(\pi z\cdot \frac{w}{\|w\|}\Big)}{\|z\|^{d+\alpha}}\text{d}z
\nonumber\\&=
     -c^{(\alpha)} \|w \|^\alpha \int_{\mathbb{R}^d}
     \frac{\sin^2(\pi z_1)}{\|z\|^{d+\alpha}}
\text{d}z
\nonumber\\&=
     -\tilde{c}^{(\alpha)} \|w \|^\alpha.
\end{align}
In the third equality we  used a change of variables. The integral is finite, since for large values of $z$ we can use that $\frac{\sin^2(\pi z_1)}{\|z\|^{d+\alpha}} \le  \frac{1}{\|z\|^{d+\alpha}}$ and for small $\frac{\sin^2(\pi z_1)}{\|z\|^{d+\alpha}} \le \frac{\pi^2}{\|z\|^{d+\alpha-2}}$. 

The best way to understand the asymptotic behaviour of $(n^{\alpha} \lambda_w^{(\alpha, n)} )_n$ is to see it as a sequence the Riemann sums which converges to the integral $\lambda^{(\alpha,\infty)}_w$ as $n\rightarrow \infty$. In general, given a function $h \in C^2(\mathbb{R}^d)$ with sufficiently fast decay at infinity, it is easy to prove the bound
\begin{equation}\label{eq-riemman-bound-c2-functions}
    \Big|\frac{1}{n^d}\sum_{x \in \frac{1}{n}\mathbb{Z}^d  } h(x) - \int_{\mathbb{R}^d} h(z)\text{d}z\Big|
\le
    c(h) \frac{1}{n}\int_{\mathbb{R}^d} \| \nabla h(z) \|\text{d}z,
\end{equation}
for some constant $c(h)>0$.
Unfortunately, this bound is not good enough for us, as
\begin{equation} \label{def-h-w}
    h_w(z):= \frac{\sin^2(\pi z \cdot w )}{\|z\|^{d+\alpha}}
\end{equation}
and its derivatives have singularities at $z=o$.

The main technical result of this section is the following proposition, which presents the necessary bounds for the inverse of the eigenvalues in the case $\alpha\in(0,2)$. The equivalent of this proposition in the case $\alpha\ge 2$ can be derived from the same techniques, but using Lemma~\ref{lem-bounds-on-eigenvalues-4} and \ref{lem-bounds-on-eigenvalues-5} instead of Lemma~\ref{lem-bounds-on-eigenvalues-1}.

\begin{proposition}\label{lem-bounds-on-eigenvalues-2}
    Let $d \ge 1$ and $\alpha \in (0,2)$ be fixed. There exists a constant 
    $C=C_{d,\alpha}>0$ such that, for all $ n \ge 1$ and
    $w \in \mathbb{Z}^d_n \backslash \{o\}$, we have
    \begin{equation}\label{lem-bounds-on-eigenvalues-2-eq-1}
        \Big|
            \frac{1}{n^\alpha\lambda^{(\alpha,n)}_w}
            -\frac{1}{\tilde{c}^{(\alpha)}\|w\|^\alpha}
        \Big|
        \le
            {C}
						\begin{cases}
						\frac{1}{n^{1-\alpha} \|w \|^{2\alpha-1}}, & \text{ if } \alpha \in (0,1) \\
                                                \frac{1}{n}\log \left (\frac{n}{\|w\|} \right ) & \text{ if } \alpha=1 \\
						\frac{1}{n^{2-\alpha} \|w \|^{2\alpha-2}}, & \text{ if } \alpha \in (1,2),
						\end{cases}
    \end{equation}
where $\tilde{c}^{(\alpha)}$ was defined in \eqref{eq:ctilde}.
\end{proposition}
The proof of the proposition is a consequence of Lemmas~\ref{lem-bounds-on-eigenvalues-0}, \ref{lem-pre-bounds-on-eigenvalues} and \ref{lem-bounds-on-eigenvalues-1} which we state and prove in the sequel.

\begin{lemma}\label{lem-bounds-on-eigenvalues-0}
    Let $d \ge 1$ and $\alpha \in (0,2)$ be fixed. There exists a
    constant $C=C_{d,\alpha}>0$ such that, for all $ n \ge 1$ and
    $w \in \mathbb{Z}^d_n \backslash \{o\}$, we have
\begin{equation}\label{lem-bounds-on-eigenvalues-0-eq-1}
    C^{-1} \frac{\|w\|^\alpha}{n^\alpha} \le -\lambda^{(\alpha,n)}_w \le C \frac{\|w\|^\alpha}{n^\alpha}.
\end{equation}
\end{lemma}
\begin{proof}
    Note that
\begin{align*}
    -\lambda^{(\alpha,n)}_w
&=
	c^{(\alpha)}
    \sum_{y \in \mathbb{Z}^d \backslash \{o\}} \frac{ \sin^2 \big(\frac{ \pi y\cdot w}{n}\big)}{\|y\|^{d+\alpha}}
\\&=
	c^{(\alpha)}
    \frac{\|w\|^\alpha}{n^\alpha}  \left(\frac{\|w\|^d}{n^d} \sum_{y \in \frac{\|w\|}{n}\mathbb{Z}^d \backslash \{o\}}
    \frac{ \sin^2(\pi y \cdot \frac{w}{\|w\|})}{\|y\|^{d+\alpha}} \right).
\end{align*}
The term in the parentheses is a Riemann sum, hence, we just need to prove that such a sum is uniformly bounded in $n$ and $w$.
Now, one proceeds by bounding the Riemann sum according to the upper and lower sum in the partition
$\left \{B \left(\frac{\|w\|}{n} y, \frac{\|w\|}{2n}\right), y \in \mathbb{Z}^d \right\}$ and noticing that upper and lower sums are monotone according to the natural partition order. Therefore,
\begin{align*}
    \frac{1}{2^d} \sum_{z \in \frac{1}{2}\mathbb{Z}^d \backslash \{o\}}
    \frac{ \sin^2(\pi z_*(z) \cdot \frac{w}{\|w\|})}{\|z_*(z)\|^{d+\alpha}}
&\le
    \frac{\|w\|^d}{n^d} \sum_{y \in \frac{\|w\|}{n}\mathbb{Z}^d \backslash \{o\}}
    \frac{ \sin^2(\pi y \cdot \frac{w}{\|w\|})}{\|y\|^{d+\alpha}}
\\&\le
    \frac{1}{2^d} \sum_{z \in \frac{1}{2}\mathbb{Z}^d \backslash \{o\}}
    \frac{ \sin^2(\pi z^*(z) \cdot \frac{w}{\|w\|})}{\|z^*(z)\|^{d+\alpha}},
\end{align*}
where
\[
	z_*(z)
	= 
	\underset{x \in B\left(z,\frac{1}{4}\right)}{\argmin} \left \{ \frac{ \sin^2(\pi x \cdot \frac{w}{\|w\|})}{\|z\|^{d+\alpha}}  \right \}
	\text{ and } 
	z^*(z)
	= 
	\underset{x \in B\left( z,\frac{1}{4}\right)}{\argmax} \left \{\frac{ \sin^2(\pi x \cdot \frac{w}{\|w\|})}{\|z\|^{d+\alpha}} \right \}.
\]
Notice that, as $z \in \frac{1}{2}\mathbb Z^d $, we have that the ball $B( z,\frac{1}{4})$ is bounded away from the origin.
Finally, one can check that both of the sums are indeed finite and positive.
\end{proof}

The following lemma will be used to prove Lemma \ref{lem-bounds-on-eigenvalues-1}, it follows from basic calculus.
Remember that we use $\|\cdot\|_D$ to denote the $L^\infty(D)$ of a function.

\begin{lemma}\label{lem-pre-bounds-on-eigenvalues}
Let $d\geq 1$, $w\in  \mathbb{Z}^d_n\backslash \{o\}$, $x \in \frac{\|w\|}{n}\mathbb{Z}^d\backslash \{o\}$ and $v \in \mathbb{R}^d$ such that $\|v\|=1$.
There exists a constant $C=C(d,\alpha)>0$ such that the norm of the gradient of
\[
	h_{v}(z)=\frac{\sin ^{2}(\pi z\cdot v)}{\|z\|^{d+\alpha}}
\]
satisfies
\[
	\left\|\nabla h_{v} (\cdot) \right\|_{B\left(x,\frac{\|w\|}{2n} \right)}
	\leq
	 C
	 \min \left\{ \frac{1}{\|x\|^{d+\alpha-1}}, \frac{1}{\|x\|^{d+\alpha}}\right \}.
\]
\end{lemma}
The last ingredient for proving Proposition~\ref{lem-bounds-on-eigenvalues-2} is the following lemma.
\begin{lemma}\label{lem-bounds-on-eigenvalues-1}
  For fixed $d \ge 1$ and $\alpha \in (0,2)$, there exists a constant 
  $C=C_{d,\alpha}>0$ such that, for all $ n \ge 1$ and
  $w \in \mathbb{Z}^d_n \backslash \{o\}$, we have
      \begin{equation}\label{lem-bounds-on-eigenvalues-1-eq-1}
          |n^\alpha\lambda^{(\alpha,n)}_w-\tilde{c}^{(\alpha)}  \|w\|^\alpha |
          \le C
					\begin{cases}
\frac{\| w\| }{n^{1-\alpha}},& \text{ if } \alpha \in (0,1) \\
						\frac{\| w\|^2 }{n} \log\left (\frac{n}{\|w\|}\right ),& \text{ if } \alpha=1\\
\frac{\| w\|^2 }{n^{2-\alpha}},& \text{ if } \alpha \in (1,2),
										\end{cases}
      \end{equation}
where $\tilde{c}^{(\alpha)}$ is defined in \eqref{eq:ctilde}.
\end{lemma}
\begin{proof}
We will study the rate of convergence of the Riemann sums of $h_{w}(z)= \frac{\sin^2(\pi z \cdot w )}{\|z\|^{d+\alpha}}$. The first step is to remove a neighbourhood around the origin.

\begin{align*}
    |n^\alpha\lambda^{(\alpha,n)}_w
&
    -\tilde{c}^{(\alpha)}  \|w\|^\alpha |
\\&=
    c^{(\alpha)}
    \Big|
    \frac{1}{n^d}
    \sum_{x \in \frac{1}{n} \mathbb{Z}^d  \backslash\{o\}}
        h_w(x)
    -
    \int_{\mathbb{R}^d }h_w(z) \text{d}z
    \Big|
\\& \le
	C
    \Big|
    \frac{1}{n^d} \sum_{x \in \frac{1}{n} \mathbb{Z}^d \backslash\{o\}} h_w(x)
    -
    \int_{\mathbb{R}^d \backslash B(o,\frac{1}{2n}) }h_w(z) \text{d}z
    \Big|
    +
    C
    \int_{B(o,\frac{1}{2n})}|h_w(z)| \text{d}z
\\& \le
	C
   \underbrace{\|w\|^\alpha \Big|
    \frac{\|w\|^d}{n^d}  \sum_{x \in  \frac{\|w\|}{n} \mathbb{Z}^d \backslash\{o\}}
        h_{\frac{w}{\|w\|} }(x)
    -
    \int_{ \mathbb{R}^d \backslash B(o,\frac{\|w\|}{2n})}
    h_{\frac{w}{\|w\|}}(z) \text{d}z
    \Big|}_{I_n(w)}
    +
    C \frac{\|w\|^2}{n^{2-\alpha}},
\end{align*}
where in the second inequality, we used that $|h_w(z)|\le  \frac{\pi^2\|w\|^2}{\|z\|^{d+\alpha-2}} $. Furthermore,
\begin{align}\label{lem-bounds-on-eigenvalues-1-eq-2}
I_n(w)& \le
    C\|w\|^\alpha
    \sum_{x \in \frac{\|w\|}{n}  \mathbb{Z}^d \backslash\{o\}}
    \int_{ B(x,\frac{\|w\|}{2n}  )}
  \left |h_{\frac{w}{\|w\|} }(x) -
    h_{\frac{w}{\|w\|}}(z) \right| \text{d}z
\nonumber\\& \le
    C\|w\|^\alpha
    \sum_{x \in \frac{\|w\|}{n}  \mathbb{Z}^d \backslash\{o\}}
    \int_{ B(x,\frac{\|w\|}{2n}  )}
    \|z-x\|
    \int_0^1 \|\nabla h_{\frac{w}{\|w\|} }(tz+(1-t)x)\|\text{d}t \text{d}z 
\end{align}
for some $C>0$. For points $z\in B(x, \frac{\|w\|}{2n} )$, we can use the bound 
\[
\left \|\nabla h_{\frac{w}{\|w\|}} (tz+(1-t)x) \right \| \le \|\nabla h_\frac{w}{\|w\|} (\cdot) \|_{B(x,\frac{\|w\|}{2n})}. 
\]
Hence, 
\[\eqref{lem-bounds-on-eigenvalues-1-eq-2}
 \le
    C\|w\|^\alpha \Big(\frac{\|w\|}{n}\Big)^{d+1}
    \sum_{x \in \frac{\|w\|}{n}\mathbb{Z}^d\backslash\{o\}}
    \|\nabla h_{\frac{w}{\|w\|}} (\cdot)\|_{B(x,\frac{\|w\|}{2n})}
\]
Using Lemma \ref{lem-pre-bounds-on-eigenvalues}, we get that the above equation can be further bounded by
\begin{equation}\label{eq:intlem2}
		I_n(w)
	\le
	C \frac{\|w\|^{1+\alpha}}{n} \underbrace{\int_{\mathbb{R}^d\backslash B(o,\frac{\|w\|}{2n})} \min \Big\{ \frac{1}{\|z\|^{d+\alpha-1}}, \frac{1}{\|z\|^{d+\alpha}}\Big\} \text{d} z}_{I'_n(w)}.
\end{equation}
To conclude the proof of the lemma, we bound the minimum in the integral above by $\frac{1}{\|\cdot\|^{d+\alpha-1}}$ for $\alpha>1$ and by 
$\frac{1}{\|\cdot\|^{d+\alpha}}$ for $\alpha<1$. Note that $\|w\| \leq C n$ for some constant $C>0$ since $w\in \mathbb{Z}^d_n$, hence the dominant term for $\alpha<1$ is $\frac{\|w\|}{n^{1-\alpha}}$. For $\alpha=1$ we write
\begin{align*}
I'_n(w) &= \int_{B(o,1)\backslash B(o,\frac{\|w\|}{2n})} \min \Big\{ \frac{1}{\|z\|^{d}}, \frac{1}{\|z\|^{d+1}}\Big\} \text{d} z + \int_{\mathbb{R}^d\backslash B(o,1)} \min \Big\{ \frac{1}{\|z\|^{d}}, \frac{1}{\|z\|^{d+1}}\Big\} \text{d} z\\
&\leq C \int^1_{\frac{\|w\|}{2n}} \frac{r^{d-1}}{r^{d}} dr + \int_{1}^{\infty} \frac{r^{d-1}}{r^{d+1}} dr \leq C \log \left( \frac{n}{\|w\|}\right),
\end{align*}
again for some $C>0$, plugging this estimate in \eqref{eq:intlem2} concludes the proof.
\end{proof}

We finish the section with two lemmas that extend Lemma~\ref{lem-bounds-on-eigenvalues-1} to $\alpha \ge 2$. The equivalent statements for the other Lemmas in this section can also be adapted. We split the between the cases $\alpha=2$ and $\alpha>2$, as the proofs use different techniques.  

\begin{lemma}\label{lem-bounds-on-eigenvalues-4}
    For fixed $d \ge 1$ and $\alpha =2$ we have 
    for all $ n \ge 1$ and $w \in \mathbb{Z}^d_n \backslash \{o\}$, for any $w\neq o$, we have
      \begin{equation} \label{lem-bounds-on-eigenvalues-4-eq-1}
	\lambda^{(2,n)}_w = -\tilde{c}^{(2)}\frac{\|w\|^2\log(n)}{n^2 }  
	  +
	    \mathcal O\left(\frac{\|w\|^2\log (\|w\|)}{n^2}\right),
      \end{equation}
      with 
      \[
	\tilde{c}^{(2)}:=  \frac{c^{(2)}\pi^{\frac{d+4}{2}}}{d\cdot\Gamma(d/2)}.
      \]
\end{lemma}

The proof borrows ideas from Lemma~\ref{lem-bounds-on-eigenvalues-1}. However, we need to keep track of some extra terms, due to the fact that $h_w$ is not integrable around the origin. In the following, we sketch how to extend the proof.

This time, instead of comparing $n^2\lambda^{(2,n)}_w$ with $\lambda^{(2,\infty)}_w$, we compare it to a second sequence $\bar{\lambda}^{(2,n)}_w$, defined as
\begin{align*}
  -\bar{\lambda}^{(\alpha,n)}_w
  & :=
  \|w\|^\alpha \int_{\frac{\|w\|}{2n}\le \|z\|<\infty} \frac{\sin^2(\pi z\cdot v)}{ \|z\|^{d+\alpha} }  \text{d}z.
\end{align*}
 One can prove that 
\begin{equation*}
  (\log (n))^{-1}\cdot \bar{\lambda}^{(2,n)}_w \longrightarrow \tilde{c}^{(2)} \|w\|^{2}.
\end{equation*}
Indeed, for $n$ sufficiently large, we can write
\begin{align*}
  (\log( n))^{-1}\cdot \bar{\lambda}^{(2,n)}_w
  & =
  -\frac{\|w\|^2}{\log(n)} \int_{\frac{\|w\|}{2n}\le \|z\|\le 1} \frac{(\pi z\cdot v)^2}{ \|z\|^{d+2} }  \text{d}z
  \\ & \phantom{==}-  
  \frac{\|w\|^2}{\log(n)}\int_{\frac{\|w\|}{2n}\le \|z\|\le 1} 
  \frac{\sin^2(\pi z\cdot v)-(\pi z\cdot v)^2}{ \|z\|^{d+2} }  \text{d}z
  \\ & \phantom{==}- 
  \frac{\|w\|^2}{\log(n)} \int_{1\le \|z\|<\infty} 
  \frac{\sin^2(\pi z\cdot v)}{ \|z\|^{d+2} }  \text{d}z
  \\ & 
  =: I_1 + I_2 +I_3.
\end{align*}
Now, using invariance of the integral $I_1$  by orthonormal transformations and computing the integral explicitly via spherical coordinates, we get
\[
   I_1
   =
   -\|w\|^2\left(\pi^2 c^{(2)}\int_{\mathbb S^{d-1}} x_1^2\text{d}\mu_{d-1}(x)\right)
   +\mathcal O\left(\frac{\|w\|^2\log(\|w\|)}{\log(n)}\right),
\]
moreover, we can evaluate the integral which is equal to $\frac{2\pi^{d/2}}{d \cdot\Gamma(d/2)}$ 
Using that $\sin^2(\pi z\cdot v)-(\pi z\cdot v)^2= O(z^4)$ in the region of integration in $I_2$, we get that
\[
  |I_2| \le C \frac{\|w\|^2}{\log(n)}.
\]
Finally, due to the integrability of $h_v$ at infinity, we have again that 
\[
  |I_3|\le C \frac{\|w\|^2}{\log(n)}.
\]
One still needs to show that $(\log n)^{-1}|n^2\lambda^{(2,n)}_w-\bar{\lambda}^{(2,n)}_w|$ is small, which is obtained by following the same strategy of Lemma~\ref{lem-bounds-on-eigenvalues-1} disregarding the region around the origin at the beginning of the argument. Moreover, we need to use the points of the grid $\frac{\|w\|}{2n}\left((\mathbb Z^d\setminus\{o\})\cap [-1,1]^d\right)$  to also control the region $B(o,\|w\|/2n)$ with Euclidean norm larger than $\|w\|/2n$.

With this, we get 
\[
  (\log(n))^{-1}\left|n^2\lambda^{(2,n)}_w-\bar{\lambda}^{(2,n)}_w\right|
  \le C  
  \frac{\|w\|^2}{n\log (n)}\int_{\|w\|/2n}^1 r^{-2} \text{d}r.
\]
Keeping track of all these contributions, we get the desired bounds. 

\begin{lemma}\label{lem-bounds-on-eigenvalues-5}
    For fixed $d \ge 1$ and $\alpha \in (2,\infty)$ we have 
    for all $ n \ge 1$ and $w \in \mathbb{Z}^d_n \backslash \{o\}$, for any $w\neq o$, we have
      \begin{equation}\label{lem-bounds-on-eigenvalues-4-eq-2} 
	\lambda^{(\alpha,n)}_w = -\tilde{c}^{(\alpha)}\frac{\|w\|^2}{n^2}
	  +
	  \begin{cases}
	    \mathcal O\left(\frac{\|w\|^4}{n^{\alpha}}\right), & \text{ if } \alpha \in (2,4)\\
	    \mathcal O\left(\frac{\|w\|^4 \log n }{n^{4}}\right)\log\left(\frac{\|w\|}{2n}\right), & \text{ if } \alpha =4\\
	    \mathcal O\left(\frac{\|w\|^{4}}{n^{4}}\right), & \text{ if } \alpha \in (4,\infty),
	  \end{cases}
      \end{equation}
      with 
      \begin{equation*} 
	\tilde{c}^{(\alpha)}:=\pi^2 c^{(\alpha)} \sum_{x \in \mathbb Z^d \setminus \{o\}}  \frac{x_1^2}{\|x\|^{d+\alpha}}.
      \end{equation*}
\end{lemma}

Notice that the constant above is $\pi^2$ times the variance of any of the coordinates of the steps of the random walk, which recovers the clear probabilistic
interpretation.
We will also restrict ourselves to only sketch the proof, which is simpler than the case $\alpha\le 2$ as it does not depend on rates of convergence of Riemann sums.
Indeed, notice that using \eqref{discrete-eigenvalues-eq},

\begin{align*}
  n^2\lambda_w^{(\alpha,n)}
  & 
  = - 
  n^2\cdot c^{(\alpha)}\sum_{\substack{x \in \mathbb{Z}^d \backslash\{o\} \\ \|x\|\le n  }}
    \frac{\sin^2(\pi x \cdot \frac{w}{n} )}{\|x\|^{d+\alpha}}
  -
  n^2\cdot c^{(\alpha)}\sum_{\substack{x \in \mathbb{Z}^d \backslash\{o\} \\ \|x\|\ge n  }}
    \frac{\sin^2(\pi x \cdot \frac{w}{n} )}{\|x\|^{d+\alpha}}
  \\ & 
  = - 
   c^{(\alpha)}\sum_{\substack{x \in \mathbb{Z}^d \backslash\{o\} \\ \|x\|\le n  }}
    \frac{(\pi x \cdot {w} )^2}{\|x\|^{d+\alpha}}
  +\mathcal{O}
  \left( \frac{\|w\|^4}{n^2} \int_{1}^n r^{3-\alpha} \text{d}r  \right)
  +
  \mathcal O \left( n^2\int_n^\infty r^{-1-\alpha} \text{d}r	\right),
\end{align*}
where in the first sum, we used a Taylor expansion of the sine function. Now, we examine the first sum and get
\[
\sum_{\substack{x \in \mathbb{Z}^d \backslash\{o\} \\ \|x\|\le n  }}
    \frac{(\pi x \cdot {w} )^2}{\|x\|^{d+\alpha}}
    =
    \|w\|^2
\sum_{x \in \mathbb{Z}^d \backslash\{o\}}
    \frac{\left(\pi x \cdot \frac{w}{\|w\|} \right)^2}{\|x\|^{d+\alpha}}
   + 
   \mathcal O\left(\frac{\|w\|^2}{n^{\alpha-2}}\right).
\]
Collecting all the terms, we get the desired error bounds.
We still need to show that the first sum does not depend on $v=w/\|w\|$,  
\begin{align*}
  \sum_{x \in \mathbb{Z}^d \backslash\{o\}}
  \frac{(\pi x \cdot {v} )^2}{\|x\|^{d+\alpha}}
  &=
  \pi^2
  \sum_{i=1}^d
  v_i^2
  \sum_{x \in \mathbb{Z}^d \backslash\{o\}}
  \frac{x_i^2}{\|x\|^{d+\alpha}}
  +
  2\pi^2\sum_{i\neq j}
  v_i v_j
  \sum_{x \in \mathbb{Z}^d \backslash\{o\}}
  \frac{x_i  x_j}{\|x\|^{d+\alpha}}.
\end{align*}
We have that $\sum_{x \in \mathbb{Z}^d \backslash\{o\}} \frac{x_i x_j}{\|x\|^{d+\alpha}}=0$ 
and  $\sum_{x \in \mathbb{Z}^d \backslash\{o\}} \frac{x_i^2}{\|x\|^{d+\alpha}}$ is finite and does not depend on the 
choice of $i$. 
Using that $\|v\|=1$ we recover the constant. 

\subsection{Proof of Theorem \ref{thm-bounds-odometer-finite}}
We will present the proof for $\alpha  \in (0,2)$, as the case $\alpha>2$ uses the same techniques, with the exception that it relies on the Lemma~\ref{lem-bounds-on-eigenvalues-5}, instead of Proposition~\ref{lem-bounds-on-eigenvalues-2}.
First note that using \eqref{char-field-long-range-divisible-eq-1}, we have for $x\in \mathbb{Z}^d_n$,
\begin{align*}
	\mathbb{E}[u_{\infty}^{\alpha}(x)]
&=
    \mathbb{E}[\eta^{\alpha}(x)] - \mathbb{E} \left [\min_{z \in \mathbb{Z}^d_n} \{ \eta^{\alpha}(z) \}\right ]
\\&=
    0 + \mathbb{E}\left[\max_{z \in \mathbb{Z}^d_n} \{\eta^{\alpha}(z)\} \right]
\end{align*}
since the field $\eta^{\alpha}$ is Gaussian and has 0 mean. Therefore, the expected odometer is equal to the expected value of the maximum of a Gaussian field. The key ingredients will be Dudley's bound \cite[Proposition 1.2.1]{Tal05} and the majorizing  measure theorem \cite[Theorem 2.1.1]{Tal05}. The idea is to study the mean of the extremes of a centred Gaussian field $(\eta^{\alpha}(x))_{x \in T}$ for some set of indexes through the metric on $T$ defined by
\begin{align}\label{def-gaussian-metric}
\nonumber
        d_\eta : T \times T &\longrightarrow \mathbb{R}
\\
        (x,y)& \longmapsto \mathbb{E}[(\eta^{\alpha}(x) -\eta^{\alpha}(y))^2]^{\frac{1}{2} }.
\end{align}
Basically, good bounds for $d_{\eta}(x,y)$ will imply good bounds for $\mathbb{E}[\max_{x \in T} \{\eta^{\alpha}(x)\}]$. 
In the sequel we will prove upper and lower bounds for $d_\eta$ for our case. Theorem~\ref{thm-bounds-odometer-finite} is a straightforward adaptation of the proofs of Propositions~8.3 and 8.8 made in \cite{Levine2016}.

\begin{proposition}\label{prop-upper-bounds-gaussian-distance}
For fixed $d \ge 1$ and $\alpha \in (0,2)$ there exists a constant $C = C_{d,\alpha}>0$ such that for all $n \ge 1$ and all $x \in \mathbb{Z}^d_n
\backslash \{ o \}$,
\[
    \mathbb{E}[(\eta^\alpha(o) - \eta^\alpha(x))^2]
\le
    C \Psi_{d,\alpha}(n,\|x\|),
\]
where
\begin{equation}\label{def-asymp-mean-u}
    \Psi_{d,\alpha}(n,r):=
    \begin{cases}
        n^{2\alpha - d-2} r^2        ,& \text{ if } \alpha > \frac{d}{2}+1\\
        \log(\frac{n}{r}) r^2    ,& \text{ if } \alpha = \frac{d}{2}+1\\
        r^{2\alpha -d}            ,& \text{ if } \alpha \in  (\frac{d}{2},\frac{d}{2}+1)\\
        \log(1+r)                    ,& \text{ if } \alpha = \frac{d}{2}\\
        1                        ,& \text{ if } \alpha < \frac{d}{2}\\
	\end{cases},
\end{equation}
for $r >0$, we define $\Psi_{d,\alpha}(n,0):=0$ for any $d,\alpha,n$.
\end{proposition}

Notice that the first two cases are only seen for $d =1$. We will split the proof in several parts. For any $x \in \mathbb{Z}^d_n \backslash \{o\}$  we have
\begin{align*}
    \mathbb{E}[(\eta^\alpha (o) - \eta^\alpha (x))^2]
&=
\sum_{ w \in \mathbb{Z}^d_n} (g^{(\alpha)}(o,w)-g^{(\alpha)}(w,x))^2
\\&=
    n^d\sum_{ w \in \mathbb{Z}^d_n} |\widehat{g^{(\alpha)}}(o,w)- \widehat{g^{\alpha}}(w,x)|^2
\\& \!\! \stackrel{\eqref{identity-fourier-of-green}}{=}
    \frac{4}{n^d}\sum_{ w \in \mathbb{Z}^d_n\backslash \{o\}} \frac{\sin^2(\frac{\pi x\cdot w}{n})}{(\lambda^{(\alpha,n)}_w)^2} = : 4 \cdot M_{n,d,\alpha}(x).
\end{align*}
One can relate the function $M_{n,d,\alpha}$ to 
\[
    G_{n,d,\alpha,x}(y)
:=
    \sum_{ w \in \mathbb{Z}^d_n\backslash \{o\}} \frac{\sin^2(\frac{\pi x\cdot w}{n})}{(\lambda^{(\alpha,n)}_w)^2} \1_{B(\frac{w}{n},\frac{1}{2n})}(y)
\]
by noticing that 
\begin{equation}
    M_{n,d,\alpha}(x) = \int_{\mathbb{R}^d}  G_{n,d,\alpha,x}(y) \text{d} y.
\end{equation}
We have the following property.
\begin{lemma}  \label{lem-equiv-G-H}
For fixed $d \ge 1$ and $\alpha \in (0,2)$, there exists a constant $C = C_{d,\alpha}>0$ such that
\[
    C^{-1} H_{n,d,\alpha,x}(y)
\le
    G_{n,d,\alpha,x}(y)
\le
    C H_{n,d,\alpha,x}(y),
\]
for $x \in \mathbb Z_n^d \backslash \{o\}$, where
\[
    H_{n,d,\alpha,x}(y)
:=
    \sum_{ w \in \mathbb{Z}^d_n\backslash \{o\}} \frac{\sin^2(\frac{\pi x\cdot w}{n})}{\|y\|^{2\alpha}} \1_{B(\frac{w}{n},\frac{1}{2n})}(y).
\]
\end{lemma}

\begin{proof}
             By the triangular inequality, we have
\begin{equation}\label{eq:trian}
    (1+\sqrt{d})^{-1} \Big\| \frac{w}{n} \Big\|
\le
    \|y\|
\le
    (1+\sqrt{d}) \Big\| \frac{w}{n}\Big\|,
\end{equation}
for $y \in B\Big(\frac{w}{n}, \frac{1}{2n}\Big), $ and $w \in \mathbb Z^d_n \backslash \{o\}$.
Therefore, using Lemma~\ref{lem-bounds-on-eigenvalues-0}, we can find constants $c_1, c_2>0$ such that
\begin{align*}
    c^{-1}_2\frac{\1_{B(\frac{w}{n},\frac{1}{2n})}(y)}{\|y\|^{2\alpha}}
&\le
    c^{-1}_1\frac{\1_{B(\frac{w}{n},\frac{1}{2n})}(y)}{(\frac{\|w\|}{n})^{2\alpha}}
\\&
\le
    \frac{\1_{B(\frac{w}{n},\frac{1}{2n})}(y)}{(\lambda^{(\alpha,n)}_w)^2}
\le
    c_1\frac{\1_{B(\frac{w}{n},\frac{1}{2n})}(y)}{(\frac{\|w\|}{n})^{2\alpha}}
\le
    c_2\frac{\1_{B(\frac{w}{n},\frac{1}{2n})}(y)}{\|y\|^{2\alpha}}.
\end{align*}
Substituting these bounds in the definition of $G_{d,n,\alpha,x}$ concludes the proof.
\end{proof}
It follows from the previous lemma that there exists a constant $C>0$ such that
\begin{equation}\label{prop-upper-bounds-gaussian-distance-bound-equiv-H}
    C^{-1} \int_{\mathbb{R}^d} H_{n,d,\alpha,x}(y) \text{d}y
\le
    \mathbb{E}[(\eta^\alpha (o) - \eta^\alpha (x))^2]
\le
    C \int_{\mathbb{R}^d} H_{n,d,\alpha,x}(y) \text{d}y.
\end{equation}
Note that the support of $H_{d,n,\alpha,x}(y)$ is contained in the annulus $B_2(o,\frac{\sqrt{d}+1}{2}) \backslash B_2(o, \frac{1}{2n} )$ hence the above integral is well-defined.
We have all the ingredients to prove Proposition~\ref{prop-upper-bounds-gaussian-distance} now. 
\begin{proof}[Proof of Proposition~\ref{prop-upper-bounds-gaussian-distance}]
We split the integral
\[
    \int_{\mathbb{R}^d} H_{n,d,\alpha,x}(y) \text{d}y
=
    \underbrace{\int_{\frac{1}{2n} <\|y\| < \frac{\sqrt{d}}{\|x\|}} H_{n,d,\alpha,x}(y) \text{d}y}_{:=I_1}
+
	\underbrace{\int_{\frac{\sqrt{d}}{\|x\|} <\|y\| < \frac{\sqrt{d}+1}{2} }  H_{n,d,\alpha,x}(y) \text{d}y}_{:=I_2},
\]
for $x \in \mathbb Z^d_n \backslash \{o\}$.
First let us look at $I_1$. Consider $y$ such that $ \frac{1}{2n} <\|y\| \le \frac{\sqrt{d}}{\|x\|} $, we use the inequality $\sin^2(t) \le t^2$, \eqref{eq:trian} and Cauchy-Schwarz to get
\begin{align}
\nonumber
             H_{n,d,\alpha,x}(y)
&\le
    \sum_{ w \in \mathbb{Z}^d_n \backslash \{o\}}\pi^2 \frac{ \|x\|^2 \|\frac{w}{n}\|^2}{\|y\|^{2\alpha}}\1_{B(\frac{w}{n},\frac{1}{2n})}(y)
\\&\le
    (1+\sqrt{d})^2\pi^2 \|x\|^2 \sum_{ w \in \mathbb{Z}^d_n \backslash \{o\}} \frac{1 }{\|y\|^{2\alpha-2}}\1_{B(\frac{w}{n},\frac{1}{2n})}(y).
\end{align}
There exists a constant $c_d>0$ such that
\begin{equation} \label{prop-upper-bounds-gaussian-distance-eq-1}
             I_1
=
    \int_{\frac{1}{2n} <\|y\| < \frac{\sqrt{d}}{\|x\|}} H_{n,d,\alpha,x}(y) \text{d}y
\le
    c_d \|x\|^2 \int _{\frac{1}{2n}}^{\frac{\sqrt{d}}{\|x\|}} r^{d+1-2\alpha} \text{d} r.
\end{equation}
On the other hand, for $y$ such that $\|y\| \in (\frac{\sqrt{d}}{\|x\|}, \frac{\sqrt{d}+ 1}{2})$, we just use the trivial bound $\sin^2(t) \le 1$. Therefore, the second integral can be bounded by
\begin{equation} \label{prop-upper-bounds-gaussian-distance-eq-2}
    I_2= \int_{\frac{\sqrt{d}}{\|x\|} <\|y\| < \frac{\sqrt{d}+1}{2} }  H_{n,d,\alpha,x}(y) \text{d}y
 \le
	 c_d \int_{\frac{ \sqrt{d}}{\|x\|}}^{\frac{\sqrt{d}+1}{2}} r^{d-1-2\alpha} \text{d}r.
\end{equation}
Computing the right-hand sides in both \eqref{prop-upper-bounds-gaussian-distance-eq-1} and  \eqref{prop-upper-bounds-gaussian-distance-eq-2}, one recovers the desired expression for $\Psi_{d,\alpha}(n,r)$.
\end{proof}
For $\alpha=2$, we can use expression \eqref{lem-bounds-on-eigenvalues-4-eq-1} to get the right estimates. In fact, one has to estimate the rate of divergence of the Riemann sums of functions $\tilde{h}_w(x)=\frac{\sin^2(\pi w \cdot x)}{\|x\|^2 \log(1/\|x\|)}$ for different dimensions involving $\log$ corrections.

For the lower bound we will distinguish different cases depending on $\alpha$ and $d$. 

\begin{lemma}\label{lem-lower-bounds-gaussian-distance-1}
    For $d=1$ and $\alpha \in (\frac{1}{2},2)$ there exists a positive constant $c=c_\alpha>0$ such that
    \[
        \mathbb{E}[(\eta^\alpha (o) - \eta^\alpha (x))^2] \ge c \Psi_{1,\alpha}(n,\|x\|),
    \]
    for all $x \in \mathbb Z_n^d\backslash \{o\}$.
\end{lemma}
\begin{proof}
We will use that  $\sin(t) \ge \frac{2}{\pi} t$ for all $t \in (0,\frac{\pi}{2})$, then
\begin{align}\label{lem-lower-bounds-gaussian-distance-1-eq-1}
\nonumber
    \mathbb{E}[(\eta^\alpha (o) - \eta^\alpha (x))^2]
&=
    M_{n,1,\alpha}(x)
\nonumber\\&=
    \frac{1}{n}\sum_{w\in\mathbb{Z}_n\backslash\{o\}}
    \frac{\sin^2(\frac{\pi x w}{n})}{(\lambda^{(\alpha,n)}_w)^2}
  \!\!\!\stackrel{\eqref{lem-bounds-on-eigenvalues-0-eq-1}}{\ge}
    c \frac{1}{n}\sum_{w\in\mathbb{Z}_n\backslash\{o\}}
    \frac{\sin^2(\frac{\pi x w}{n})}{\frac{\|w\|^{2\alpha}}{n^{2\alpha}}}
\nonumber\\ & \ge
    c \frac{1}{n}\sum_{\substack{w\in\mathbb{Z}_n\backslash\{o\}\\ \|w\| < \frac{n}{2\|x\|}}}
    \frac{\frac{\|x\|^2 \|w\|^2}{n^2}}{\frac{\|w\|^{2\alpha}}{n^{2\alpha}}} =
    c n^{2\alpha-3} \|x\|^2\sum_{\substack{w\in\mathbb{Z}_n\backslash\{o\}\\ \|w\| < \frac{n}{2\|x\|}}}\|w\|^{2-2\alpha},
\end{align}
then one recovers the right estimates by evaluating the sum, which will either be convergent ($\alpha > \frac{ 3}{2}$), diverges logarithmically ($\alpha= \frac{3}{2}$) or diverges polynomially ($\frac{1}{2}< \alpha < \frac{3}{2}$).
\end{proof}

\begin{lemma}\label{lem-lower-bounds-gaussian-distance-2}
    For $d\in\{2,3\}$ and $\alpha \in (\frac{d}{2},2)$ there exists a positive constant $c=c_{d,\alpha}>0$, we have
    \[
	    \mathbb{E}[(\eta^\alpha (o) - \eta^\alpha (x))^2] \ge c \Psi_{d,\alpha}(n,\|x\|),
    \]
    for $x\in \mathbb{Z}^d_n\backslash \{o\}$.
\end{lemma}
\begin{proof}
Let $S_k  = \mathbb{Z}^d \cap \partial B(0,k)\subset \mathbb{Z}^d$.
For any $x \in \mathbb{R}^d$, we define $H_x = \{y \in \mathbb{Z}^d: |x \cdot y |\ge \frac{1}{\sqrt{d}}\|x\|\|y\|\}$.
One can easily check there is a positive constant $c_d>0$ such that $|S_k \cap H_x|\ge c_d k^{d-1}$ for all $k \ge 1$ and all $x \in \mathbb{R}^d$. Let $a \in (1,\sqrt{d})$, if $\|w\|\le \frac{n}{a \|x\|}$, then by Cauchy-Schwarz, we have $|x \cdot w|n^{-1} \le a^{-1}$.
Now, we use the inequality $b|t| \le |\sin(\pi\cdot t)| \le \pi^2 |t|$ for all $|t|\le a^{-1}$ with $b:= a \sin (\pi a^{-1})$.
Hence, we get
\begin{align}\label{lem-lower-bounds-gaussian-distance-2-eq-1}
    M_{n,d,\alpha}(x)
& \ge
    \frac{1}{n^d}
    \sum_{k=1}^{\big\lfloor  \frac{n}{a \|x\|}\big\rfloor \wedge \left\lfloor \frac{n}{4}\right\rfloor  }
    \sum_{ w \in S_k}
    \frac{\sin^2(\frac{\pi x\cdot w}{n})}{(\lambda^{(\alpha,n)}_w)^2}
 \ge
    \frac{c_d}{n^d}
    \sum_{k=1}^{\big\lfloor  \frac{n}{a \|x\|}\big\rfloor \wedge \left\lfloor \frac{n}{4}\right\rfloor  }
    \sum_{ w \in S_k}
    \frac{\frac{|w\cdot x|^2}{n^2}}{\frac{\|w\|^{2\alpha}}{n^{2\alpha}}}
\nonumber\\& \ge
    c_d n^{2\alpha-d-2} \|x\|^2
    \sum_{k=1}^{\big\lfloor  \frac{n}{a \|x\|}\big\rfloor \wedge \left\lfloor \frac{n}{4}\right\rfloor  }\sum_{ w \in S_k\cap H_x}
    \|w\|^{2-2\alpha}
 \nonumber\\
	   &\ge
    c_d n^{2\alpha-d-2} \|x\|^2
    \sum_{k=1}^{\big\lfloor  \frac{n}{a \|x\|}\big\rfloor \wedge \left\lfloor \frac{n}{4}\right\rfloor}
    k^{d-1}k^{2-2\alpha}
\nonumber\\& \!\!\!\!\stackrel{(\alpha > d/2)}{\ge}
    c_d n^{2\alpha-d-2} \|x\|^2
    \Big({\Big\lfloor  \frac{n}{a \|x\|}\Big\rfloor \wedge \left\lfloor \frac{n}{4}\right\rfloor}\Big)^{d+2-2\alpha}.
\end{align}
As  $\alpha > \frac{d}{2}$ and $a <\sqrt{d}$, the right-hand side of \eqref{lem-lower-bounds-gaussian-distance-2-eq-1} is of order $\|x\|^{2\alpha-d} = \Psi_{d,\alpha}(n,\|x\|)$.
\end{proof}

For the case $d > 2 \alpha$ and $x \neq o$ we have to analyse the rate of convergence of the function
$H_{n,d,\alpha,x}(y)$ to its almost everywhere pointwise limit, that is
\[
    H_{\infty,d,\alpha,x}(y)
=
    \frac{\sin^2(\pi y \cdot x)}{\|y\|^{2\alpha}} \1_{B\left(o, \frac{1}{2}\right)\backslash\{o\}}(y).
\]
In particular, for $d\ge 2$, it will be useful to express
\begin{equation}\label{lower-bounds-gaussian-distance-projection}
	\int_{r_1 \le \|y\|\le r_2} H_{\infty,d,\alpha,x}(y) \text{d}y
=
	\int_{r_1}^{r_2} \frac{v_d(r\|x\|)}{r^{2\alpha}} r^{d-1} \text{d}r,
\end{equation}
where $v_d(t) := \int_{\mathbb{S}^{d-1}} \sin^2(\pi t y_1) \mu_{d-1}(\text{d} y)$ and $\mu_{d-1}$ is the surface measure on the sphere $\mathbb{S}^{d-1}$.

\begin{lemma}\label{lem-lower-bounds-gaussian-distance-3}
For $d \ge 2$, and for all $\varepsilon >0$, there exists $\delta >0$ such that if $t \ge \varepsilon$,  then $v_d(t)\ge \delta$.
\end{lemma}

\begin{proof}
The case $d\ge 3$ is covered in Lemma~8.4,  \cite{Levine2016}. For $d=2$, we need to prove that $\underline{\lim}_{t \to \infty} v_d(t) >0$. By using \cite[Corollary 4]{Bar05}, we obtain

\[
    v_2(t)
=
    c_2 \int_{-1}^1 (1- y^2)^{-\frac{1}{2}} \sin^2(\pi t y) \text{d} y
\ge
    \frac{ 2 c_2}{\sqrt{3}}\int_{-\frac{1}{2}}^{\frac{1}{2}} \sin^2(\pi t y) \text{d}y.
\]
Now, the result follows by noticing that $\lim_{t \to \infty}\int_{-\frac{1}{2}}^{\frac{1}{2}} \sin^2(\pi t y) \text{d}y = \frac{1}{2}$.
\end{proof}

The proofs of the next two lemmas are equivalent to the proofs of Lemma~8.5 and 8.6 in \cite{Levine2016}.

\begin{lemma}\label{lem-lower-bounds-gaussian-distance-4}
Let $d \ge 1$. For all $\varepsilon >0$, there exist $\delta,N >0$ such that
\[
    \Big |H_{n,d,\alpha,x}(y) -  \frac{\sin^2(\pi y \cdot x)}{\|y\|^{2\alpha}} \Big |
\le
    \frac{\varepsilon}{\|y\|^{2\alpha}}
\]
for all $n \ge N$, $x \in \mathbb{Z}^d_n\backslash \{o\}$ such that $\|x\| \le \delta n$ and for almost every $y$ in the annulus $ B_2(o,\frac{1}{4}) \backslash B_2(o,\frac{1}{8 \|x\|} )$.
\end{lemma}

\begin{lemma}\label{lem-lower-bounds-gaussian-distance-5}
Let $d \ge 2$ and $\alpha \in (0,2)$ such that $\alpha \le \frac{d}{2}$. There exist $\delta,N,c_{d,\alpha} >0$ such that
\[
    \int_{\mathbb{R}^d} H_{n,d,\alpha,x}(y) \text{d}y
\ge
    c_{d,\alpha} \int_{\frac{1}{8\|x\|}}^{\frac{1}{4}} r^{d-2\alpha-1} \text{d}r
\]
for all $n \ge N$, $x \in \mathbb{Z}^d_n\backslash \{o\}$ satisfying  $\|x\| \le \delta n$.
\end{lemma}

We are left with the case $d=1$ and $\alpha \leq \frac{1}{2}$. Here we have to compute the lower bound of  $\int_{\mathbb{R}^d} H_{n,d,\alpha,x}(y) \text{d}y$ directly as we cannot apply the same ideas as in the proofs of Lemma~\ref{lem-lower-bounds-gaussian-distance-1} or Lemma~\ref{lem-lower-bounds-gaussian-distance-3}. 

\begin{lemma}\label{lem-lower-bounds-gaussian-distance-6}
Let $d = 1$ and  $\alpha \in (0,\frac{1}{2}]$. There exists $\delta,N,c_{1,\alpha} >0$ such that
\[
  \int_{\mathbb{R}} H_{n,1,\alpha,x}(y) \text{d}y
\ge
  c_{1,\alpha} \Psi_{1,\alpha}(n,\|x\|)
\]
for all $n \ge N$, $x \in \mathbb{Z}_n\backslash \{o\}$ satisfying $\|x\| \le \delta n$.
\end{lemma}

\begin{proof}
Let $\varepsilon_1>0$ to be chosen later. By Lemma~\ref{lem-lower-bounds-gaussian-distance-4}, we can find positive constants $\delta, N>0$ such that
\[
    \Big |H_{n,1,\alpha,x}(y) -  \frac{\sin^2(\pi y \cdot x)}{\|y\|^{2\alpha}} \Big |
\le
    2\frac{\varepsilon_1}{\|y\|^{2\alpha}}
\]
for all $n \ge N$ and for all $x$ such that $\|x\|\le \delta n$ and for almost every $y$ in the annulus $ B_2(0, \frac{1}{4}) \backslash B_2(0, \frac{1}{8\|x\|})$.

Therefore, for $n \ge N$ and $x$ such that $\|x\|\le \delta n$, we have
\begin{align*}
    \int_{\mathbb{R}} H_{n,1,\alpha,x} (y) \text{d} y
&\ge
    \int_{\frac{1}{8\|x\|}\le \|y\| \le \frac{1}{4}} H_{n,1,\alpha,x} (y) \text{d} y
\\&\ge
    \int_{\frac{1}{8\|x\|}\le \|y\| \le \frac{1}{4}}
    (H_{\infty,1,\alpha,x} (y)-\varepsilon_1 2 \|y\|^{-2\alpha} )\text{d} y
\\&\ge
  \underbrace{ 2\int_{\frac{1}{8\|x\|}\le r \le \frac{1}{4}}
    r^{-2\alpha}(\sin^2(\pi x r) - \varepsilon_1)\text{d} r}_{I_{\alpha}(x)}
\\ & =
    2\int_{0 \le r \le \frac{1}{4}}
    r^{-2\alpha}\sin^2(\pi x r)\text{d} r
-
    2\varepsilon_1\int_{0 \le r \le \frac{1}{4}}
    r^{-2\alpha} \text{d} r
\\ & \qquad -
    2\int_{0 \le r \le \frac{1}{8\|x\|}}
    r^{-2\alpha}(\sin^2(\pi x r) - \varepsilon_1)\text{d} r.
\end{align*}

We first discuss the case $\alpha < \frac{1}{2}$.
The last integral converges to $0$ as $x \longrightarrow \infty$, the second is finite and the first is bounded below by a positive constant. Hence, if one chooses $\varepsilon_1$ small enough, one can show that the sum of the three integrals is bounded below by a positive constant (uniform in $n$ and $x$). For $\alpha=\frac{1}{2}$, we have that
\[
I_{\alpha}(x) \geq C \log(1 + \|x\|)
\]
where $C>0$ is some positive constant for $\varepsilon_1>0$ small enough.
\end{proof}

\begin{proposition}\label{prop-lower-bounds-gaussian-distance}
Let $d \ge 1$ and  $\alpha \in (0,2)$. There exist $\delta, N, C >0$ such that
\[
    \mathbb{E}[(\eta^\alpha(o) - \eta^\alpha(x))^2]
\ge
    C^{-1}\Psi_{d,\alpha}(n,\|x\|)
\]
for all $n \ge N$, $x \in \mathbb{Z}^d_n\backslash\{o\}$ satisfying $\|x\| \le \delta n$ and $\Psi_{d,\alpha}$ defined as in
\eqref{def-asymp-mean-u}.
\end{proposition}

The proof of the proposition is a combination of Lemmas~\ref{lem-lower-bounds-gaussian-distance-1}, \ref{lem-lower-bounds-gaussian-distance-2}, \ref{lem-lower-bounds-gaussian-distance-5} and \ref{lem-lower-bounds-gaussian-distance-6}.

\subsection{Proof of Theorem \ref{theorem-main-non-Gaussian}}
\label{subsec-proof-non-gaussian}
We will only present the proof for $\alpha \in (0,2)$ as the general case follows in a similar way. It will be divided into two parts (analogously to the proof of Theorem~2 in \cite{Cipriani2016}):
\begin{enumerate}
    \item We first prove convergence of finite dimensional distributions to the field $\Xi^{\alpha}$, that is  $\{(\Xi^{\alpha}_n,f)\}_{f \in \mathcal{F}}$ converges to  $\{(\Xi^{\alpha},f)\}_{f \in \mathcal{F}}$ for any finite collection $\mathcal{F}$ of test functions in the appropriate space.
    \item Secondly, we prove tightness of the law of  $\Xi^{\alpha}_n$. We will take advantage of a classical result given by Theorem~\ref{thm-rellich} which characterises compact embedding  of Sobolev spaces.
    \end{enumerate}

The main difference between the proof of Theorem~\ref{theorem-main-non-Gaussian} and Theorem~2 in \cite{Cipriani2016} is the asymptotics of the eigenvalues of $-(-\Delta)^{\sfrac{\alpha}{2}}_n$. In \cite{Cipriani2016}, the authors use Lemma 7 to bound the eigenvalues of the discrete Laplacian (up to the correct renormalisation) and with respect to its continuous counterpart. In particular, their lower-bound can be taken uniformly. However, in our case such bounds cannot be obtained in the same way. We rely on the asymptotic behaviour of the eigenvalues of $-(-\Delta)^{\sfrac{\alpha}{2}}_n$, as described throughout Subsection~\ref{subsec-estimates-eigenvalues}.

Moreover, once the comparison between the rescaled eigenvalues of the discrete fractional Laplacian and its continuous version is established, the rest of the proof follows easily for large values of $\alpha$ ($\alpha >1 $). However, for small values of $\alpha$ ($\alpha <1$ and in particular $\alpha<1/2$), the technical bounds necessary to make use of the dominated convergence theorem in the proof of finite-dimensional distributions has to be evaluated with more care. The rest of the proof follows in a similar way, with the analogous adaptations. However, we include its proof to keep the article more self-contained.

Note that, for all $k\ge1$ and $\theta_1, \cdots, \theta_k \in \mathbb{R}$, $f^{(1)}, \cdots, f^{(k)} \in C^\infty(\mathbb{T}^d)$,
\[
    ( \Xi_n^{\alpha}, \theta_1 f^{(1)} + \cdots + \theta_k f^{(k)} )
    \stackrel{d}{=}
    \theta_1  (\Xi_n^{\alpha}, f^{(1)})
    +\cdots+
    \theta_k (\Xi_n^{\alpha},f^{(k)}).
\]
Hence,  looking at the characteristic function we have
\begin{align*}
    \phi_{((\Xi_n^{\alpha},  f^{(1)}), \cdots, (\Xi_n^{\alpha}, f^{(k)}))}(\theta_1,\cdots, \theta_k)
&=
    \mathbb{ E}\Bigg[\exp\Bigg( i\Big( \theta_1 (\Xi_n^{\alpha}, f^{(1)})
    +\cdots+
    \theta_k (\Xi_n^{\alpha},f^{(k)}) \Big)\Bigg) \Bigg]
\\&=
    \phi_{(\Xi_n^{\alpha}, \theta_1 f^{(1)} + \cdots + \theta_k f^{(k)})}(1),
\end{align*}
therefore, it will be enough to study the distribution of a single coordinate of the field, that is $( \Xi_n^{\alpha},f )$.
By Proposition~\ref{prop:rescaledGauss} the odometer can be represented as
\begin{equation}
    u^{\alpha}_{\infty}(x)
\stackrel{d}{=}
    \eta^{\alpha}(x) - \min_{z \in \mathbb{Z}^d_n} \{ \eta^{\alpha}(z)\},
\end{equation}
for each $x\in \mathbb{Z}^d_n$, where
\begin{equation}\label{def-aux-subs-odometer}
\begin{split}
    \eta^{\alpha}(x)&= \sum_{y \in \mathbb{Z}^d_n}  g^{(\alpha)}(x,y)(s(y)-1)  \\
& = w_n(x) -\frac{1}{n^d}\sum_{y \in \mathbb{Z}^d_n }  g^{(\alpha)}(x,y)
    \sum_{z \in \mathbb{Z}^d_n }\sigma(z),
\end{split}
\end{equation}
where
\begin{equation*}
  w_n(x) :=
  \sum_{y \in \mathbb{Z}^d_n}  g^{(\alpha)}(x,y)\sigma(y).
\end{equation*}

Given a function $h_n: \mathbb{Z}^d_n \longrightarrow \mathbb{R}$, one can define
\begin{align*}
	\Xi^{\alpha}_{h_n}(x)
:=
    \tilde{c}^{(\alpha)} \sum_{z \in \mathbb{T}^d_n}
    n^{\frac{d-2\alpha}{2}}h_n(nz)\1_{B(z,\frac{1}{2n})}(x),
\qquad
    x \in \mathbb{T}^d
\end{align*}
and recall that we denoted by $\Xi^{\alpha}_n$ the field corresponding to $h_n=u^{\alpha}_{\infty}$ defined in \eqref{def-aproaching-field-eq} and $\tilde{c}^{(\alpha)}$ defined in \eqref{eq:ctilde}.
Then, for $f\in C^{\infty}(\mathbb{T}^d)$ such that
$\int_{\mathbb{T}^d}f(z) \text{d}z=0$, we have that
\[
	( \Xi^{\alpha}_{n}, f )  =
 	( \Xi^{\alpha}_{w_n}, f ),
\]
since the last sum in \eqref{def-aux-subs-odometer} is invariant and does not depend on $x$.
We prove convergence of all moments of $( \Xi^{\alpha}_{w_n}, f )$ first for $\sigma$'s for which all moments exist and then for the general case. 

\subsubsection{Convergence for weights with finite moments}
\label{subsec-proof-conv-bounded-weights}
In this section, we will prove the following theorem.

\begin{theorem}
    Assume that $(\sigma(x))_{x \in \mathbb{Z}^d_n}$ is a collection of i.i.d
    random variables such that  $\mathbb{E}[\sigma(x)]=0,~ \mathbb{E}[\sigma^2(x)]=1$
    and   $ \mathbb{E}[|\sigma(x)|^k] < \infty$ for all $k \in \mathbb{N}$. Let $d\ge 1$
    and $u^{\alpha}_{\infty}$ the odometer for the long-range divisible sandpile in
    $\mathbb{Z}^d_n $. Then the field $\Xi^{\alpha}_n$ defined
    in \eqref{def-aproaching-field-eq} converges
    weakly to $\Xi^{\alpha}$ as $n\rightarrow \infty$. The convergence holds in the same manner as in
    Theorem~\ref{theorem-main-non-Gaussian}.
\end{theorem}

First let us prove the following proposition.
\begin{proposition}\label{prop-convergence-moments-bounded-case}
    Assume $\mathbb{E}[\sigma(x)]=0,~ \mathbb{E}[\sigma^2(x)]=1$ and that
    $ \mathbb{E}[|\sigma(x)|^k] < \infty$ for all $k \in \mathbb{N}$ and $x\in \mathbb{Z}^d_n$. Then for all $m \ge 1$
    and for all $f \in C^\infty (\mathbb{T}^d)$ with zero mean, the following
    limit holds:
    \begin{equation}
        \lim_{n \to \infty} \mathbb{E}[(\Xi^{\alpha}_{w_n}, f )^m ]
        =
        \begin{cases}
            (2m-1)!!\|f \|_{-\frac{\alpha}{2}}^m, & m \in 2\mathbb{N} \\
            0, & m \in 2\mathbb{N} +1.
        \end{cases}
    \end{equation}
\end{proposition}
\begin{proof}
For $f \in C^\infty (\mathbb{T}^d)$ define the map $T_n: \mathbb{T}^d \longrightarrow \mathbb{R}$ by

\begin{equation}\label{def-function-T}
	 z \longmapsto \int_{B(z, \frac{1}{2n})}    f(y) \text{d} y.
\end{equation}

\noindent \textbf{Case} $m=2$: We have the equality
\begin{align*}
    \mathbb{E}[w_n(y)w_n(y^\prime)]
&=
    \sum_{x \in \mathbb{Z}^d_n }     g^{(\alpha)}(y,x)
    \sum_{x^\prime \in \mathbb{Z}^d_n } g^{(\alpha)}(x^\prime,y^\prime)
    \mathbb{E}[\sigma(x)\sigma(x^\prime)]
\\&=
    \sum_{x \in \mathbb{Z}^d_n }     g^{(\alpha)}(y,x)
    g^{(\alpha)}(x,y^\prime).
\end{align*}
This implies that
\[
    \mathbb{E} [(\Xi^{\alpha}_{w_n},f )^2 ]
=
    (\tilde{c}^{(\alpha)})^2 n^{d-2\alpha}\sum_{x \in \mathbb{Z}^d_n}
    \Big( \sum_{z \in \mathbb{T}^d_n }     g^{(\alpha)}(x,nz)T_n(z)\Big)^2.
\]
We now use that, analogously to the proof of Proposition 4 in \cite{Cipriani2016},
\begin{equation}\label{prop-convergence-moments-bounded-case-eq-1}
\begin{split}  
  \sum_{x \in \mathbb{Z}^d_n}     g^{(\alpha)}(x,y)g^{(\alpha)}(x,y^\prime)
& = n^d \widehat{g^{\alpha}}(o,y) \widehat{g^{\alpha}}(o,y') + n^d \sum_{x\in \mathbb{Z}^d_n \backslash \{o\}} \widehat{g^{\alpha}}(x,y) \widehat{g^{\alpha}}(x,y')\\
    &=n^dL^2 + C^{(\alpha)}_n(y,y^\prime),
\end{split}
\end{equation}
where $L=\widehat{g^{\alpha}}(o,\cdot)$ is constant. The term $n^d L^2$ can be dealt with by defining a common Gaussian random variable, independent of the rest of the field, with mean zero and variance $n^d L^2$.
This common random variable will not matter as we are restricting ourselves to mean zero functions.
The second part of \eqref{prop-convergence-moments-bounded-case-eq-1} can be written as
\begin{equation}\label{def-of-covariance}
        C^{(\alpha)}_n(x,y) := 
        \frac{1}{n^d} \sum_{z \in \mathbb{Z}^d_n \backslash \{0\}}
        \frac{\exp(2\pi i (y-x)\cdot\frac{z}{n})}{(\lambda^{(\alpha,n)}_{z})^2}.
\end{equation}
Hence,
\begin{align*}
    \mathbb{E} [(\Xi^{\alpha}_{w_n},f )^2 ]
&=
    (\tilde{c}^{(\alpha)})^2 n^{d-2\alpha}\sum_{z,z^\prime \in \mathbb{T}^d_n}
    C^{(\alpha)}_n(nz,nz^\prime) T_n(z)T_n(z^\prime)
\\&=
    (\tilde{c}^{(\alpha)})^2 n^{d-2\alpha}\sum_{z,z^\prime  \in \mathbb{T}^d_n}
    C^{(\alpha)}_n(nz,nz^\prime)
    \int_{B(z,\frac{1}{2n} )}
\!\!\!\!\!
    f(x) \text{d}x
    \int_{B(z^\prime,\frac{1}{2n} )}
\!\!\!\!\!
    f(x^\prime) \text{d}x^\prime.
\end{align*}
Our strategy will be to divide the above sum in three parts:
\begin{align*}
    \mathbb{E} [(\Xi^{\alpha}_{w_n},f )^2 ]
&=
	(\tilde{c}^{(\alpha)})^2 n^{-d-2\alpha}\sum_{z,z^\prime \in \mathbb{T}^d_n}
    C^{(\alpha)}_n(nz,nz^\prime)     f(z)f(z^\prime)
\\&+
    (\tilde{c}^{(\alpha)})^2 n^{-d-2\alpha}\sum_{z,z^\prime \in \mathbb{T}^d_n}
    C^{(\alpha)}_n(nz,nz^\prime)K_n(f)(z)K_n(f)(z^\prime)
\\&+
	2(\tilde{c}^{(\alpha)})^2 n^{-d-2\alpha}\sum_{z,z^\prime \in \mathbb{T}^d_n}
    C^{(\alpha)}_n(nz,nz^\prime) f(z)K_n(f)(z^\prime),
\end{align*}
where $K_n$ is defined as
\begin{equation}\label{def-error-K}
    K_n(f)(z):=n^d \Big[ \int_{B(z,\frac{1}{2n})} (f(x)-f(z)) \text{d}x\Big].
\end{equation}
Using Propositions~\ref{prop-convariances-converge} and \ref{prop-rest-is-negligible} and Cauchy-Schwarz inequality we will prove that
\[
    \lim_{n \to \infty} \mathbb{E} [(\Xi^{\alpha}_{w_n},f)^2] =
    \|f\|^2_{-\frac{\alpha}{2}},
\]
concluding the proof for the case $m=2$.
\begin{proposition}\label{prop-convariances-converge} For any
$f \in C^{\infty}(\mathbb{T}^d)$ with
$\int_{\mathbb{T}^d} f(x) \text{d} x =0$, we have that
\begin{align*}
    \lim_{n \to \infty}
    (\tilde{c}^{(\alpha)})^2
    n^{-d-2\alpha}
    \sum_{z,z^\prime \in \mathbb{T}^d_n}f(z)f(z^\prime)
    C^{(\alpha)}_n(nz,nz^\prime)
&=
    \| f\|^2_{-\frac{\alpha}{2}}.
\end{align*}
\end{proposition}

\begin{proposition}\label{prop-rest-is-negligible}
For any $f\in C^{\infty}(\mathbb{T}^d)$ with
$\int_{\mathbb{T}^d} f(x) \text{d} x =0$,
    \[
        \lim_{n \to \infty}
	(\tilde{c}^{(\alpha)})^2 n^{-d-2\alpha}\sum_{z,z^\prime \in \mathbb{T}^d_n}
        C^{(\alpha)}_n(nz,nz^\prime)K_n(f)(z)K_n(f)(z^\prime)
    =
        0.
    \]
\end{proposition}

We will first prove Proposition~\ref{prop-rest-is-negligible}, it is an easy consequence of the following lemma.

\begin{lemma}\label{lem-aux-rn-vanishes-in-l2-1}
    There exists a constant $C >0$ such that
    $\sup_{z \in \mathbb{T}^{d}} |K_{n}(f)(z)| \leq \frac{C}{n}$.
\end{lemma}
\begin{proof}
    Using the mean value inequality, we have that there exists
    $c_{x,z} \in (0,1)$ such that
    \begin{align*}
        |K_{n}(f)(z)|
    & \le
        n^{d} \int_{B (z,\frac{1}{2n})}|f(x)-f(z)|\text{d} x
    \\ & \le
        n^{d} \int_{B\ (z,\frac{1}{2n})} \|\nabla f(c_{x,z}x+(1-c_{x,z})z)\|\|z-x\| \text{d} x
    \\ & \le
        C \frac{n^d}{2n}
        \int_{B\ (z,\frac{1}{2n})} \|\nabla f(c_{x,z}x+(1-c_{x,z})z)\|\text{d}x \le
        \frac{C}{2n}  \|\nabla f (\cdot)\|_{\mathbb{T}^d}.
    \end{align*}
    The lemma follows from the fact that
    $\|\nabla f (\cdot)\|_{\mathbb{T}^d} < \infty$.
\end{proof}

\begin{proof}[Proof of Proposition~\ref{prop-rest-is-negligible}]
Let $K^\prime_n(f)(z)=K_n(f)(\frac{z}{n})$ and write
\begin{align*}
&
 (\tilde{c}^{(\alpha)})^2 n^{-2d} \sum_{z,z^{\prime} \in \mathbb{T}_{n}^d }n^{d-2\alpha}
    C^{(\alpha)}_n(nz,nz^{\prime})K_{n}(f)(z)K_{n}(f)(z^{\prime})
\\& = (\tilde{c}^{(\alpha)})^2 n^{-2d} \sum_{z,z^{\prime} \in \mathbb{T}_{n}^d }
     \sum_{w \in \mathbb{Z}^d_n \backslash \{o\}}
        \frac{\exp(2\pi i (z-z^\prime)\cdot w )}{
        (n^{\alpha}\lambda^{(\alpha,n)}_{w})^2}
        K_{n}(f)(z)K_{n}(f)(z^{\prime})
\\& \!\!\! \!\!\! \! \!\stackrel{\text{Lemma}~\ref{lem-bounds-on-eigenvalues-0}}\le
    c
    \sum_{w \in \mathbb{Z}^d_n\backslash\{o\}}
    |\widehat{K^\prime_n(f)}(w)|^2,
\end{align*}
where, we used that $\alpha > 0$ and $\|w\|\ge 1$. Notice that
\begin{align*}
    \sum_{w \in \mathbb{Z}^d_n\backslash\{o\}} |\widehat{K^\prime_n(f)}(w)|^2
&\le
    \sum_{w \in \mathbb{Z}^d_n}|\widehat{K^\prime_n(f)}(w)|^2 \le n^{-d}
    \sum_{w \in \mathbb{Z}^d_n} |K^\prime_n(f)(w)|^2
\\ & \le
    \|K_n(f)\|_{\mathbb{T}^d}^2 \le
    \frac{C}{n^2}.
\end{align*}
This completes the proof of Proposition~\ref{prop-rest-is-negligible}.
\end{proof}
\begin{proof}[Proof of Proposition~\ref{prop-convariances-converge}]
To prove Proposition~\ref{prop-convariances-converge} we will rely on information about the speed of convergence of the eigenvalues $\lambda^{(\alpha,n)}_w$, proven in Proposition~\ref{lem-bounds-on-eigenvalues-2} in Subsection~\ref{subsec-estimates-eigenvalues}. Notice that

\begin{align}
\nonumber\label{prop-convariances-converge-eq-1}
    \lim_{n \to \infty}
&
    n^{-d-2\alpha}
    (\tilde{c}^{(\alpha)})^2\sum_{z,z^\prime \in \mathbb{T}^d_n}f(z)f(z^\prime)
    C^{(\alpha)}_n(nz,nz^\prime)
\nonumber\\&=
    \lim_{n \to \infty}  n^{-2d}
    (\tilde{c}^{(\alpha)})^2\sum_{z,z^\prime \in \mathbb{T}^d_n}f(z)f(z^\prime)
    \sum_{w \in \mathbb{Z}^d_n \backslash \{o\}}
    \frac{\exp(2\pi i (z-z^\prime)\cdot w )
    }{(n^\alpha \lambda^{(\alpha,n)}_{w})^2}
\nonumber\\&  =
    \lim_{n \to \infty}  n^{-2d}
    (\tilde{c}^{(\alpha)})^2\sum_{z,z^\prime \in \mathbb{T}^d_n}f(z)f(z^\prime)
    \Big(
    \sum_{w \in \mathbb{Z}^d_n \backslash \{o\}}
    \frac{\exp(2\pi i (z-z^\prime)\cdot w )
    }{(\tilde{c}^{(\alpha)})^2 \|w\|^{2\alpha}}
\nonumber\\&\qquad+
    2\exp(2\pi i (z-z^\prime)\cdot w )
    \Big(
        \frac{1}{n^\alpha\lambda^{(\alpha,n)}_w}-\frac{1}{\tilde{c}^{(\alpha)}\|w\|^\alpha}
    \Big)\frac{1}{\tilde{c}^{(\alpha)}\|w\|^\alpha}
\nonumber\\&\qquad+
\nonumber        \exp(2\pi i (z-z^\prime)\cdot w ) \Big(
        \frac{1}{n^\alpha\lambda^{(\alpha,n)}_w}-\frac{1}{\tilde{c}^{(\alpha)}\|w\|^\alpha}
    \Big)^2\Big) \\
& = I + II + III.
\end{align}

However, we will show that the last two summands  are irrelevant. Recall that $f_n: \mathbb{Z}^d_n \longrightarrow \mathbb{R} $ was defined as
$f_n(z) = f (\frac{z}{n} ) $. First we will prove that the third term is irrelevant.

\noindent {\underline{\textit{Case $\alpha \in (1, 2) $:}}} We have that

\begin{align*}
     \Big| (\tilde{c}^{(\alpha)})^2 n^{-2d}&
\!\!
    \sum_{z,z^\prime \in \mathbb{T}^d_n}
\!\!
    f(z)f(z^\prime)
\!\!
    \sum_{w \in \mathbb{Z}^d_n \backslash \{o\}}
    \exp \Big(2\pi i (z-z^\prime)\cdot w \Big)
    \Big(
    \frac{1}{n^\alpha\lambda^{(\alpha,n)}_w}-\frac{1}{\tilde{c}^{(\alpha)}\|w\|^\alpha}
    \Big)^2 \Big|
\\& \!\!\!\stackrel{\eqref{lem-bounds-on-eigenvalues-2-eq-1}}{\le}
   c \frac{1}{n^{4-2\alpha}}\sum_{w \in \mathbb{Z}^d_n \backslash \{0\}}
    \frac{|\widehat{f_n}(w)|^2 }{\|w \|^{4\alpha-4}}
\\& \overset{\|w\| \geq 1}\le
    c \frac{1}{n^{4-2\alpha}}
    \sum_{w \in \mathbb{Z}^d_n }
    |\widehat{f_n}(w)|^2=
   c \frac{1}{n^{4-2\alpha}}
     \frac{1}{n^d} \sum_{w \in  \mathbb{T}^d_n }
    |f(w)|^2
\end{align*}
where in the last equality we used Parseval's identity.
As
$\frac{1}{n^d} \sum_{w \in \frac{1}{n} \mathbb{Z}^d_n }|f(w)|^2
\longrightarrow \int_{\mathbb{T}^d} |f(z)|^2 \text{d}z$,  the last term in the
above expression vanishes as $n \longrightarrow \infty$.

\noindent {\underline{\textit{Case $\alpha \in (\frac{1}{2}, 1)$:}}} The proof follows analogously to the
previous one, we look at the third term in the brackets of expression
\eqref{prop-convariances-converge-eq-1} to get

\begin{align*}
    \Big|(\tilde{c}^{(\alpha)})^2 n^{-2d}&
\!\!
    \sum_{z,z^\prime \in \mathbb{T}^d_n}
\!\!
    f(z)f(z^\prime)
\!\!
    \sum_{w \in \mathbb{Z}^d_n \backslash \{o\}}
    \exp\Big(2\pi i (z-z^\prime)\cdot w \Big)
    \Big(
    \frac{1}{n^\alpha\lambda^{(\alpha,n)}_w}-\frac{1}{\tilde{c}^{(\alpha)}\|w\|^\alpha}
    \Big)^2 \Big|
\\& \!\!\!\stackrel{\eqref{lem-bounds-on-eigenvalues-2-eq-1}}{\le}
    c \frac{1}{n^{2-2\alpha}}\sum_{w \in \mathbb{Z}^d_n \backslash \{o\}}
    \frac{|\widehat{f_n}(w)|^2 }{\|w \|^{4\alpha-2}} \le
    c \frac{1}{n^{2-2\alpha}}
     \frac{1}{n^d} \sum_{w \in  \mathbb{T}^d_n }
    |f(w)|^2 \longrightarrow 0,
\end{align*}
as $n\rightarrow \infty$ using the same reasoning as before.

\noindent {\underline{\textit{Case $\alpha \in (0,\frac{1}{2}]$:}}} In this case we write the third term of
\eqref{prop-convariances-converge-eq-1} as

\begin{align*}
 \Big|(\tilde{c}^{(\alpha)})^2 n^{-2d}&
\!\!
    \sum_{z,z^\prime \in \mathbb{T}^d_n}
\!\!
    f(z)f(z^\prime)
\!\!
    \sum_{w \in \mathbb{Z}^d_n \backslash \{o\}}
    \exp\Big(2\pi i (z-z^\prime)\cdot w \Big)
    \Big(
    \frac{1}{n^\alpha\lambda^{(\alpha,n)}_w}-\frac{1}{\tilde{c}^{(\alpha)}\|w\|^\alpha}
    \Big)^2 \Big|
\\& \!\!\!\stackrel{\eqref{lem-bounds-on-eigenvalues-2-eq-1}}{\le}
    c \frac{1}{n^{2-2\alpha}}\sum_{w \in \mathbb{Z}^d_n \backslash \{o\}}
        \frac{|\widehat{f_n}(w)|^2}{\|w \|^{4\alpha-2}} \le
    c \frac{1}{n^{2\alpha}}
     \frac{1}{n^d} \sum_{w \in  \mathbb{T}^d_n }
    |f(w)|^2 \longrightarrow 0,
\end{align*}
where, in the last inequality, we used that for $w \in \mathbb{Z}^d_n $, $\|w\|^{2-4\alpha} \le C n^{2-4\alpha}$ together with Parseval's identity. 
For $\alpha=1$ we compute
\[
\begin{split}
III & \leq    c \frac{1}{n^{2}}\sum_{w \in \mathbb{Z}^d_n \backslash \{o\}}
        |\widehat{f_n}(w)|^2 \log^2 \left ( \frac{n}{\|w\|}\right ) \leq c \frac{\log^2(n)}{n^{2}}
     \frac{1}{n^d} \sum_{w \in  \mathbb{T}^d_n }
    |f(w)|^2 \longrightarrow 0,
\end{split}
\]
as $n\rightarrow \infty$. This proves that the third term of the summand inside the brackets in \eqref{prop-convariances-converge-eq-1} vanishes as $n\rightarrow \infty$. To prove that the second term II in \eqref{prop-convariances-converge-eq-1} vanishes, one can proceed in the in a similar manner, distinguishing the cases $\alpha \in (1,2)$, $\alpha \in (1/3,1)$, and $\alpha \in (0,1/3)$ and then considering the special cases $\alpha=\frac{1}{3}$ and $\alpha=1$. In fact, for $\alpha=\frac{1}{3}$ we have
\[
\begin{split}
II  &\leq    c \frac{1}{n^{2/3}}\sum_{w \in \mathbb{Z}^d_n \backslash \{o\}}
        |\widehat{f_n}(w)|^2  \leq c \frac{1}{n^{2/3}}
     \frac{1}{n^d} \sum_{w \in  \mathbb{T}^d_n }
    |f(w)|^2 \longrightarrow 0,
\end{split}
\]
and for $\alpha=1$
\[
\begin{split}
II & \leq    c \frac{1}{n}\sum_{w \in \mathbb{Z}^d_n \backslash \{o\}}
        \frac{|\widehat{f_n}(w)|^2}{\|w\|} \log \left ( \frac{n}{\|w\|}\right )  \overset{\|w\|\geq 1}\leq c \frac{\log(n)}{n}
     \frac{1}{n^d} \sum_{w \in  \mathbb{T}^d_n }
    |f(w)|^2 \longrightarrow 0.
\end{split}
\]
It remains to prove that for all $\alpha \in (0,2)$, we have
\begin{equation} \label{eq-finite-dim-convergence-gaussian-1}
    \lim_{n \to \infty} (\tilde{c}^{(\alpha)})^2 n^{-2d}
\!\!\!
    \sum_{z,z^\prime \in \mathbb{T}^d_n}f(z)f(z^\prime)
    \sum_{w \in \mathbb{Z}^d_n \backslash \{0\}}
    \frac{\exp(2\pi i (z-z^\prime)\cdot w )}{\|w\|^{2\alpha}}
    =
    \|f\|^2_{-\frac{\alpha}{2}}.
\end{equation}
We will distinguish different cases, depending on dimension $d$ and $\alpha$, for which
$\sum_{x \in\mathbb{Z}^d \backslash\{o\}} \|x\|^{-2\alpha}$ is convergent or not.\\
\noindent
\underline{\textit{Case $d < 2 \alpha$}}:  In this simple case  we have that
\begin{align*}
    n^{-2d}
&
    \sum_{z,z^\prime \in \mathbb{T}^d_n}f(z)f(z^\prime)
    \sum_{w \in \mathbb{Z}^d_n \backslash \{o\}}
    \frac{\exp(2\pi i (z-z^\prime)\cdot w )}{\|w\|^{2\alpha}}
\\&=
    \sum_{w \in \mathbb{Z}^d }
    \frac{\1_{w \in \mathbb{Z}^d_n\backslash \{o\}}}{\|w\|^{2\alpha}}
    \sum_{z \in \mathbb{T}^d_n} \frac{f(z) \exp(2\pi i z \cdot w )}{n^d}
    \sum_{z^\prime \in \mathbb{T}^d_n}
    \frac{f(z^\prime)\exp(-2\pi i z^\prime\cdot w)}{n^d},
\end{align*}
applying the dominated convergence theorem (notice the uniform bound as $f$ is bounded on the torus $\mathbb{T}^d $), we get
\eqref{eq-finite-dim-convergence-gaussian-1}.

\noindent
\underline{\textit{Case $d \ge 2 \alpha$}}: Here we need to make use of mollifiers.
Let $\rho \in C^\infty(\mathbb{R}^d)$ a positive function in the Schwartz space
with support in $[-\frac{1}{2},\frac{1}{2} )^d$ and
satisfying $\int_{\mathbb{R}^d} \rho(x) \text{d}x=1$, let
$\rho_\kappa(x): =
\frac{1}{\kappa^d} \rho\Big(\frac{x}{\kappa}  \Big)$ with $\kappa>0$ . Theorem 7.22 from \cite{rudin1991functional} yields that for any $m \in \{0,1,2,\dots\}$ there exists $C= C(\kappa,m)>0$ such that
\begin{equation} \label{eq-decay-of-mollifier}
        \Big| \widehat{\rho_\kappa}(w)\Big| \le \frac{C}{(1 + \|w\|)^m}.
\end{equation}
We will prove in the following that the convergence in
\eqref{eq-finite-dim-convergence-gaussian-1} is equivalent to the convergence of
\begin{equation}\label{eq-mollified-goes-to-norm}
    \lim_{\kappa \to 0^+} \lim_{ n \to \infty}
   (\tilde{c}^{(\alpha)})^2 n^{-2d}
    \sum_{z,z^\prime \in \mathbb{T}^d_n}f(z)f(z^\prime)
    \sum_{w \in \mathbb{Z}^d_n \backslash \{o\}} \widehat{\rho_\kappa}(w)
    \frac{\exp(2\pi i (z-z^\prime)\cdot w )}{\|w\|^{2\alpha}}=
    \|f\|^2_{-\frac{\alpha}{2}}.
\end{equation}
To do so, we will show that
\begin{equation}\label{eq-error-in-mollification}
    \lim_{\kappa \to 0^+} \overline{\lim_{ n \to \infty}}
    \Big|
    n^{-2d}
\!\!\!
    \sum_{z,z^\prime \in \mathbb{T}^d_n}
\!\!\!\!
    f(z)f(z^\prime)
\!\!\!\!\!
    \sum_{w \in \mathbb{Z}^d_n \backslash \{o\}}
\!\!\!\!\!
    \Big(1-\widehat{\rho_\kappa}(w)\Big)
    \frac{\exp(2\pi i (z-z^\prime)\cdot w )}{\|w\|^{2\alpha}}
    \Big|
    =0.
\end{equation}
Since $\int_{\mathbb{R}^d} \rho_{\kappa}(x)\text{d}x=1$, we have that
\[
    |\widehat{\rho_{\kappa}}(w)-1|
    \leq
    \int_{\mathbb{R}^d}\rho_{\kappa}(y)|e^{2\pi i y\cdot w}-1|dy.
\]
Moreover from $|\exp(2\pi i x)-1|^{2}=4\sin^{2}(\pi x)$ and
$|\sin(x)|\leq|x|$ we obtain
\begin{align} \label{eq-error-fourier-of-mollifier}
    |\widehat{\rho_{\kappa}}(w)-1|
    \leq
    C\kappa\|w\|\int_{\mathbb{R}^{d}}\|y\|\rho(y)dy
    \leq
    C\kappa\|w\|.
\end{align}
Therefore,
\begin{align*}
    \Big|
    n^{-2d}
&
    \sum_{ w \in \mathbb{Z}_{n}^{d}\backslash\{o\}}
    \frac { \widehat { \rho_\kappa } (w)-1} { \| w \| ^ {2\alpha} }
    \sum_{z,z^{\prime}\in\mathbb{T}_{n}^{d}}f(z)f(z^{\prime})
    \exp\Big(2\pi i \Big(z-z^{\prime}\Big)\cdot w\Big)
    \Big|
\\ & \le
    C\kappa
    \sum_{ w \in \mathbb{Z}_{n}^{d}\backslash\{o\}}
     \| w \| ^ {1-2\alpha}
        |\widehat{f_n}(w)|^2.
\end{align*}
For $\alpha \ge \frac{1}{2}$, as $\|w\|\ge 1$, we have
\begin{align*}
    \sum_{ w \in \mathbb{Z}_{n}^{d}\backslash\{o\}}
    \| w \| ^ {1-2\alpha} |\widehat{f_n}(w)|^2
& \le
    \sum_{ w \in \mathbb{Z}_{n}^{d}\backslash\{o\}}
    |\widehat{f_n}(w)|^2
 \le
    \frac{1}{n^d} \sum_{ z \in \mathbb{T}_{n}^{d}} |f(z)|^2
\end{align*}
where we used Parseval's identity, as before. Hence, we have
\begin{align*}
    \overline{\lim_{ n \to \infty}} &
    \Big|
    n^{-2d}
    \sum_{z,z^\prime \in \mathbb{T}^d_n}f(z)f(z^\prime)
    \sum_{w \in \mathbb{Z}^d_n \backslash \{o\}}
    \Big(1-\widehat{\rho_\kappa}(w)\Big)
    \frac{\exp(2\pi i (z-z^\prime)\cdot w )}{\|w\|^{2\alpha}}
    \Big|
\\ & \le
    C \kappa \|f (\cdot)\|^2_{\mathbb{T}^d},
\end{align*}
which proves \eqref{eq-error-in-mollification} letting $\kappa$ go to 0.
For the case $\alpha < \frac{1}{2}$, we use the bound
\begin{align} \label{eq-error-fourier-of-mollifier-2}
    |\widehat{\rho_{\kappa}}(w)-1|
    \leq
    C \min\{\kappa\|w\|,1\},
\end{align}
for an appropriate constant that does not depend on $\kappa$ nor on $w$. So
we can repeat the approach
\begin{align*}
    \Big|
    n^{-2d}
&
    \sum_{ w \in \mathbb{Z}_{n}^{d}\backslash\{o\}}
    \frac { \widehat { \rho_\kappa } (w)-1} { \| w \| ^ {2\alpha} }
    \sum_{z,z^{\prime}\in\mathbb{T}_{n}^{d}}f(z)f(z^{\prime})
    \exp\Big(2\pi i \Big(z-z^{\prime}\Big)\cdot w\Big)
    \Big|
\\ & \le
    C
    \sum_{ w \in \mathbb{Z}_{n}^{d}\backslash\{o\}}
     \min\{\kappa\| w \| ^ {1-2\alpha}, \|w\|^{-2\alpha}\}
        |\widehat{f_n}(w)|^2
\\ & \le
    C
    \sum_{\substack {w \in \mathbb{Z}_{n}^{d}\backslash\{o\}\\
    \|w\|\le \frac{1}{\kappa} }}
     \kappa \underbrace{\| w \| ^ {1-2\alpha}}_{\le \frac{1}{\kappa^{1-2\alpha}} }
        |\widehat{f_n}(w)|^2
    +
    C
    \sum_{\substack {w \in \mathbb{Z}_{n}^{d}\backslash\{o\}\\
    \|w\|\ge \frac{1}{\kappa} }}
    \underbrace{\|w\|^{-2\alpha}}_{\le \kappa^{2\alpha}}|\widehat{f_n}(w)|^2
\\ & \le
    C\kappa^{2\alpha}
    \sum_{w \in \mathbb{Z}_{n}^{d}\backslash\{o\}}
        |\widehat{f_n}(w)|^2    ,
\end{align*}
where in the last inequality we used that $\alpha < \frac{1}{2}$ and recover
\eqref{eq-error-in-mollification}.
The proof will be complete once we show \eqref{eq-mollified-goes-to-norm}. We will apply the dominated convergence theorem twice. First note that, as $\widehat{\rho_{\kappa}}$ decays fast at infinity, 
\begin{align*}
    \lim_{ n \longrightarrow \infty } n ^ {-d}
    \sum_{z\in\mathbb{T}_{n}^{d } } f (z)\exp(2\pi i z\cdot w)
&=
    \widehat {f}(w),
\end{align*}
and
\begin{align*}
    \lim_{ n \to \infty}&
    (\tilde{c}^{(\alpha)})^2 n^{-2d}
    \sum_{z,z^\prime \in \mathbb{T}^d_n}f(z)f(z^\prime)
    \sum_{w \in \mathbb{Z}^d_n \backslash \{o\}} \widehat{\rho_\kappa}(w)
    \frac{\exp(2\pi i (z-z^\prime)\cdot w )}{\|w\|^{2\alpha}}
\\&=
(\tilde{c}^{(\alpha)})^2    \sum_{w \in \mathbb{Z}^d \backslash \{o\}} \widehat{\rho_\kappa}(w)
    \frac{|\widehat{f}(w)|^2}{\|w\|^{2\alpha}}.
\end{align*}
As $|\widehat{\rho_{\kappa}}(\cdot)|\le 1$, we get the desired equation \eqref{eq-finite-dim-convergence-gaussian-1}.That concludes the proof of  Proposition~\ref{prop-convariances-converge}.
\end{proof}

\vspace{1em}
\noindent {\textbf{Case $m\ge 3$.}} We still need to prove Proposition~\ref{prop-convergence-moments-bounded-case} for higher order moments, however this will be a much easier result as we can now rely on Propositions~\ref{prop-convariances-converge} and \ref{prop-rest-is-negligible}. We will also need this auxiliary Lemma~12 from \cite{Cipriani2016}.

\begin{lemma}\label{lemma-control-fourier-map-T}
    Let $f \in C^\infty (\mathbb{T}^d)$ with mean zero, $T_n$ specified in \eqref{def-function-T} and $\mathcal{T}_n: \mathbb{Z}^d_n\longrightarrow \mathbb{R}$ defined by $\mathcal{T}_n(z):= T_n(\frac{z}{n})$.
		Then there exists a positive constant $\mathcal{M}= \mathcal{M}(d,f) < \infty$  such that
    \[
        n^{d} \sum_{z \in \mathbb{Z}^d_n}|\widehat{\mathcal{T}_n}(z)| \le \mathcal{M}.
    \]
\end{lemma}
For $m\in\{1,2,\dots\}$, define $\mathcal{P}(m)$
the set of partitions of $\{1,2,\dots,m\}$. Moreover, denote by $\Pi$ the elements of a
partition $P \in \mathcal{P}(m)$. We will denote $|\Pi|$ the number of elements in
$\Pi$.  Call $\mathcal{P}_2(m)\subset  \mathcal{P}(m)$ the pair
partitions, that is, partitions $P \in \mathcal{P}(m)$ such that for all
$\Pi \in P$, $|\Pi|=2$. We obtain
\begin{align}\label{prop-convergence-moments-bounded-case-eq-2}
&
    \mathbb{E}
    [( \Xi^{\alpha}_{w_n},
    f )^m ]
=
    (\tilde{c}^{(\alpha)}n^{\frac{d-2\alpha}{2} })^m
    \sum_{z_1,\dots,z_n \in \mathbb{T}^d_n } \mathbb{E}
    \Big[ \prod_{j=1}^m w_n(nz_j)\Big] \prod_{j=1}^m T_n(z_j)
\nonumber \\&=
    (\tilde{c}^{(\alpha)}n^{\frac{d-2\alpha}{2} })^m
    \sum_{P \in \mathcal{P}(m) } \prod_{\Pi \in P}
    \mathbb{E}[\sigma^{|\Pi|}(x)]
    \sum_{x \in \mathbb{Z}^d_n}\Big(
        \sum_{z_j \in \mathbb{T}^d_n: j \in \Pi }
        \prod_{j \in \Pi}g^{(\alpha)}(x,nz_j)T_n(z_j)
    \Big)
\nonumber \\&=
    \sum_{P \in \mathcal{P}(m) } \prod_{\Pi \in P}
    (\tilde{c}^{(\alpha)}n^{\frac{d-2\alpha}{2} })^{|\Pi|}
    \mathbb{E}[\sigma^{|\Pi|}(x)]
    \sum_{x \in \mathbb{Z}^d_n}\Big(
        \sum_{z \in \mathbb{T}^d_n}
        g^{(\alpha)}(x,nz)T_n(z)
    \Big)^{|\Pi|}
\end{align}
For a fixed $P$, let us consider in the product $\Pi \in P$ any term corresponding
to a block $\Pi$ with $|\Pi|=1$, this will give no contribution to the sum as
$\sigma$ have mean zero. Now consider $\Pi \in P$ with $k := |\Pi|> 2$. We have
that

\begin{align*}
    (\tilde{c}^{(\alpha)}n^{\frac{d-2\alpha}{2} })^k
&
    \mathbb{E}[\sigma^k(x)]
    \sum_{x \in \mathbb{Z}^d_n}\Big(
        \sum_{z \in \mathbb{T}^d_n}
        g^{(\alpha)}(x,nz)T_n(z)
    \Big)^{k}
\\&=
    (\tilde{c}^{(\alpha)}n^{\frac{d-2\alpha}{2} })^k
    \mathbb{E}[\sigma^k(x)]
    \sum_{x \in \mathbb{Z}^d_n}\Big(
        \sum_{z \in \mathbb{Z}^d_n}
        g^{(\alpha)}(x,z)\mathcal{T}_n(z)
    \Big)^{k}.
\end{align*}
Now we apply Parseval's identity and get

\begin{align}\label{prop-convergence-moments-bounded-case-eq-3}
    (\tilde{c}^{(\alpha)}n^{\frac{d-2\alpha}{2}})^k
&
    \mathbb{E}[\sigma^k(x)]
    \sum_{x \in \mathbb{Z}^d_n}\Big(
        n^d\sum_{z \in \mathbb{Z}^d_n}
        \widehat{g^{(\alpha)}}(x,z)\widehat{\mathcal{T}_n}(z)
    \Big)^{k}
\nonumber \\ & \stackrel{\eqref{identity-fourier-of-green}}{=}
    (\tilde{c}^{(\alpha)}n^{\frac{d-2\alpha}{2} })^k
    \mathbb{E}[\sigma^k(x)]
    \sum_{x \in \mathbb{Z}^d_n}\Big(
        \sum_{z \in \mathbb{Z}^d_n \backslash \{0\}}
        \frac{\psi_{-z}(x)}{-\lambda^{(\alpha,n)}_{z}} \widehat{\mathcal{T}_n}(z)
    \Big)^{k}.
\end{align}
We used that $\widehat{\mathcal{T}}_n(0)=0$. Now, we evoke
Lemma~\ref{lem-bounds-on-eigenvalues-0} and the fact that $\|w\|\ge 1$ to obtain that
$-\lambda^{(\alpha,n)}_{w} \ge c n^{-\alpha}$ for all $w \in \mathbb{Z}^d_n$. Therefore,
the above expression is bounded from above by
\begin{align}
    (\tilde{c}^{(\alpha)}n^{\frac{d-2\alpha}{2} })^k
    \mathbb{E}[\sigma^k(x)] &
    \sum_{x \in \mathbb{Z}^d_n}\Big(
        \sum_{z \in \mathbb{T}^d_n}
        g^{(\alpha)}(x,nz)T_n(z)
    \Big)^{k}
\nonumber\\&\le
    C n^{\frac{d k}{2} + d}
    \mathbb{E}[\sigma^k(x)]
    \Big(
        \sum_{z \in \mathbb{Z}^d_n} \Big|\widehat{\mathcal{T}}_n(z)\Big|
    \Big)^{k}.
\end{align}
As the moments of $\sigma$ are finite, we can use Lemma
\ref{lemma-control-fourier-map-T} to bound the term in parentheses above. Hence,
each block of cardinality $k>2$ has order at most
 $n^{\frac{k d}{2}-(k-1)d}=o(1)$. Therefore, the only terms of
 \eqref{prop-convergence-moments-bounded-case-eq-2} that contribute as
 $n \longrightarrow \infty$ are the ones with $k=2$, only the pair partitions.
Since $\mathcal{P}_2(2m+1) = \emptyset$, the
 odd moments will vanish. Therefore,
 \[
     \mathbb{E}\Big[( \Xi^{\alpha}_{w_n},f )^{2m} \Big] =
     \sum_{P \in \mathcal{P}_2(2m) }
     \Big((\tilde{c}^{(\alpha)})^2 n^{d-2\alpha}
    \sum_{x \in \mathbb{Z}^d_n}\Big(
        \sum_{z \in \mathbb{Z}^d_n}
        g^{(\alpha)}(x,z)\mathcal{T}_n(z)
    \Big)^{2}\Big)^m +o(1).
 \]
Note that $|\mathcal{P}_2(2m)|=(2m-1)!!$ and that the bracket term above
 converges to $\|f\|^2_{-\frac{\alpha}{2}}$. This concludes the proof of Proposition~\ref{prop-convergence-moments-bounded-case}
 and with it, we conclude point $(1)$ in the strategy of the proof (described at the beginning of Subsection~\ref{subsec-proof-non-gaussian})
\end{proof}

\textit{Tightness:} For proving tightness we will need the following result which is proven in Theorem 5.8 in \cite{roe2013elliptic}.

\begin{theorem}[Rellich's theorem]\label{thm-rellich}
    If $k_1<k_2$ the inclusion operator
    $H^{k_2}(\mathbb{T}^d)\hookrightarrow H^{k_1}(\mathbb{T}^d)$
    is a compact linear operator.
    In particular for any radius $R>0$, the closed ball
    $\overline{B_{\mathcal H_{-\frac{\varepsilon}{2}}}(0,\,R)}$ is compact in
    $\mathcal H_{-\varepsilon}$.
\end{theorem}
Choose $-\varepsilon<-\frac{d}{2}$. Observe that
\[
    \|{\Xi^{\alpha}_{w_n}}\|_{L^2(\mathbb{T}^d)}^2
    =(\tilde{c}^{(\alpha)})^2 n^{d-2\alpha}\sum_{x,\,y\in \mathbb{Z}^d_n}
    g^{(\alpha)}(x,y)\sigma(x)
    \sum_{x',\,y'\in \mathbb{Z}^d_n }g^{(\alpha)}(x',y')\sigma(x')
\]
is a.s. finite, as, for any fixed $n$, it is a finite combination of random variables in $L^2$.
Therefore, $\Xi^{\alpha}_{w_n}\in L^2(\mathbb{T}^d) \subset
\mathcal H_{-\varepsilon}(\mathbb{T}^d)$  a.s. Due to  Rellich's Theorem, it is enough
to show that, for all $\delta>0$, there exists a constant $R=R(\delta)>0$ such that
\[
     \sup_{n\in \mathbb{N}}
     \mathbb{P}
     \Big(\|\Xi^{\alpha}_{w_n}\|_{\mathcal H_{-\frac{\varepsilon}{2}}}\ge R\Big)\le \delta.
 \]
However, one can use Markov's inequality to show that it is enough to get
\[
 \sup_{n\in \mathbb{N}} \mathbb{E}
 \Big[\|\Xi^{\alpha}_{w_n}\|_{\mathcal H_{-\frac{\varepsilon}{2}}}^2\Big]
 \le C
\]
for some constant $C>0$.
Since  $\Xi^{\alpha}_{w_n}\in L^2(\mathbb{T}^d)$, we get a representation
$\Xi^{\alpha}_{w_n}(z)=\sum_{\nu\in \mathbb{Z}^d}\widehat{\Xi^{\alpha}_{w_n}}(\nu)
\mathbf \phi_\nu(z)$ in terms of eigenfunctions, we use the notation $\widehat{\Xi^{\alpha}_{w_n}}(\nu):=(\Xi^{\alpha}_{w_n},\, \phi_\nu)$.
Thus, we can express
\begin{align*}
    \|\Xi^{\alpha}_{w_n}\|_{\mathcal H_{-\frac{\varepsilon}{2}}}^2
&=
    \sum_{\nu\in \mathbb{Z}^d \backslash \{o\}}\|\nu\|^{-2\varepsilon}
    \Big|\widehat{\Xi^{\alpha}_{w_n}}(\nu)\Big|^2.
\end{align*}
Note that
\[
\widehat{\Xi^{\alpha}_{w_n}}(\nu)
=
    \int_{\mathbb{T}^d }\Xi^{\alpha}_{w_n}(z)
    \phi_{\nu}(z) d z
=
   \tilde{c}^{(\alpha)}
    \sum_{x\in \mathbb{T}^d_n}n^{\frac{d-2\alpha}{2}} w_n(nx)
    \int_{B(x,\,\frac{1}{2n})} \phi_{\nu}(z) d z.
\]
This gives
\begin{align}\label{proof-of-tighness-eq-1}
    \mathbb{E}
&
    \Big[\|\Xi^{\alpha}_{w_n}\|_{\mathcal H_{-\frac{\varepsilon}{2}}}^2\Big]
\nonumber\\&=
( \tilde{c}^{(\alpha)})^2\sum_{\nu\in \mathbb{Z}^d \backslash\{o\}}
    \sum_{x,y\in \mathbb{T}^d_n}\|\nu\|^{-2\varepsilon}n^{d-2\alpha}
    \mathbb{E} \Big[w_n(nx)w_n(ny)\Big]
    \int_{B(x,\,\frac{1}{2n})} \phi_{\nu}(z)\text{d} z
    \int_{B(y,\,\frac{1}{2n})}\overline{ \phi_{\nu}(z)}\text{d}z
\\&\stackrel{\eqref{prop-convergence-moments-bounded-case-eq-1}}{=}
\!\!\!\!\!\!
 (\tilde{c}^{(\alpha)})^2
      \sum_{\nu\in \mathbb{Z}^d \backslash\{o\}}\sum_{x,\,y\in \mathbb{T}^d_n}
    \frac{n^{d-2\alpha}}{\|\nu\|^{2\varepsilon}} \Big(n^d L^2+C^{(\alpha)}_n(nx,\,ny)\Big)
    \int_{B(x,\,\frac{1}{2n})}
\!\!\!\!
    \phi_{\nu}(z)\text{d}z
    \int_{B(y,\,\frac{1}{2n})}
\!\!\!\!
    \overline{\phi_{\nu}(z)}\text{d}z.
\end{align}
Since $\int_{\mathbb{T}^d}\phi_{\nu}(z) \text{d}z=0$, the previous expression
reduces to
\[
   (\tilde{c}^{(\alpha)})^2
    \sum_{\nu\in \mathbb{Z}^d\backslash\{o\}}
    \sum_{x,\,y\in \mathbb{T}^d_n}\|\nu\|^{-2\varepsilon}n^{d-2\alpha}
    C^{(\alpha)}_n(nx,\,ny)
    \int_{B(x,\frac{1}{2n})}\phi_{\nu}(z)\text{d}z
    \int_{B(y,\,\frac{1}{2n})}\overline{\phi_{\nu}(z)}\text{d}z.
\]
Define $F_{n,\nu}:\mathbb{T}^d_n \rightarrow \mathbb{C}$ as the function
$F_{n,\nu}(x):=\int_{B(x,\frac{1}{2n})} \phi_\nu(z)d z$. Since
$ \phi_\nu \in L^2(\mathbb{T}^d)$, by Cauchy-Schwarz inequality we get
$F_{n,\nu}\in L^1(\mathbb{T}^d_n)$. Now, we claim that
\begin{claim}\label{claim-bound-for-tightness}
    There exists $C^\prime>0$ such that
\begin{equation}\label{eq-claim-bound-for-tightness-1}
    \sup_{\nu\in \mathbb{Z}^d}\sup_{n\in\mathbb{N} }
    \sum_{x,\, y\in \mathbb{T}^d_n}n^{d-2 \alpha}C^{(\alpha)}_n(nx,\,ny)
    F_{n,\nu}(x)\overline{F_{n,\nu}(y)}
    \le C^\prime.
\end{equation}
\end{claim}
Supposing that the claim above is valid, we have that
\begin{align*}
    \mathbb{E} \Big[\|\Xi^{\alpha}_{w_n}\|_{\mathcal
    H_{-\frac{\varepsilon}{2}}}^2\Big]
&=
    (\tilde{c}^{(\alpha)})^2
\!\!\!
    \sum_{\nu\in \mathbb{Z}^d \backslash\{o\}}\|\nu\|^{-2\varepsilon}
\!\!\! 
    \sum_{x, \, y\in \mathbb{T}^d_n}n^{d-2\alpha}
    C^{(\alpha)}_n(nx,\,ny)F_{n,\nu }(x)\overline{F_{n,\nu}(y)}
\\&\le
    C'\sum_{k\ge 1}k^{d-1-2\varepsilon}  \le
    C.
\end{align*}
It remains to prove the Claim~\ref{claim-bound-for-tightness}.

\begin{proof}[Proof of Claim~\ref{claim-bound-for-tightness}]
Again, we will rely on the bounds of Lemma~\ref{lem-bounds-on-eigenvalues-0},
we will also use that
\[
    \sum_{x,\,y\in \mathbb{T}^d_n} \exp(2\pi i (x-y)\cdot w)
    F_{n,\nu}(x) \overline{F_{n,\nu}(y)}
=
    \Big|\widehat{F_{n,\nu}}(w)\Big|^2 n^{2d}\ge 0.
\]
Now, we analyse
\begin{align}\label{claim-bound-for-tightness-eq-1}
&
    \sum_{x,\,y\in \mathbb{T}^d_n}n^{d-2\alpha}C^{(\alpha)}_n(nx,\,ny)
    F_{n,\nu}(x)\overline{F_{n,\nu}(y)}
\nonumber\\&=
    \sum_{x,\,y\in \mathbb{T}^d_n} \sum_{w \in \mathbb{Z}^d_n \backslash\{o\}}
    \frac{\exp(2 \pi i (x-y)\cdot w)}{\left(n^\alpha \lambda^{(\alpha,n)}_w\right)^2}
    F_{n,\nu}(x)\overline{F_{n,\nu}(y)}
\nonumber\\& \!\!\!\!\!\!\!\!  \le
    c \sum_{x,\,y\in \mathbb{T}^d_n} \sum_{w \in \mathbb{Z}^d_n \backslash\{o\}}
    \frac{\exp(2 \pi i (x-y)\cdot w)}{\|w\|^{2\alpha}}
    F_{n,\nu}(x)\overline{F_{n,\nu}(y)}.
\end{align}
Where in the last inequality, we are using Lemma~\ref{lem-bounds-on-eigenvalues-0}.
Again consider  mollifiers $\rho_\kappa$ as before. We rewrite the sum in right-hand
side of \eqref{claim-bound-for-tightness-eq-1} as

\begin{align}\label{claim-bound-for-tightness-eq-2}
    &\sum_{x,\,y\in \mathbb{T}^d_n} \sum_{w \in \mathbb{Z}^d_n \backslash\{o\}}
    \widehat{\rho_\kappa}(w) \frac{\exp(2 \pi i (x-y)\cdot w)}{\|w\|^{2\alpha}}
    F_{n,\nu}(x)\overline{F_{n,\nu}(y)}
\nonumber\\&+
    \sum_{x,\,y\in \mathbb{T}^d_n} \sum_{w \in \mathbb{Z}^d_n \backslash\{o\}}
    \Big( 1- \widehat{\rho_\kappa}(w) \Big)\frac{\exp(2 \pi i (x-y)\cdot w)}{\|w\|^{2\alpha}}
    F_{n,\nu}(x)\overline{F_{n,\nu}(y)}.
\end{align}
In the sequel we will bound the two summands independently, starting with the second.
Consider $G_{n,\nu}: \mathbb{Z}^d_n \longrightarrow \mathbb{C} $, given by
$G_{n,\nu}(z):= F_{n,\nu}(\frac{z}{n})$. We have
\begin{align*}
&
    \sum_{x,\,y\in \mathbb{T}^d_n} \sum_{w \in \mathbb{Z}^d_n \backslash\{o\}}
    \Big( 1- \widehat{\rho_\kappa}(w) \Big)
    \frac{\exp(2 \pi i (x-y)\cdot w)}{\|w\|^{2\alpha}}
    F_{n,\nu}(x)\overline{F_{n,\nu}(y)}
\\& =
     \sum_{w \in \mathbb{Z}^d_n \backslash\{0\}}
    \Big( \frac{1- \widehat{\rho_\kappa}(w)}{\|w\|^{2\alpha}} \Big)
    \sum_{x,\,y\in \mathbb{Z}^d_n}
    \exp \Big(2 \pi i (x-y)\cdot \frac{w}{n} \Big)
    F_{n,\nu}\Big(\frac{x}{n} \Big)\overline{F_{n,\nu}\Big(\frac{y}{n}\Big)}
\\& =
     \sum_{w \in \mathbb{Z}^d_n \backslash\{0\}}
    \Big( \frac{1- \widehat{\rho_\kappa}(w)}{\|w\|^{2\alpha}} \Big)
    \widehat{G_{n,\nu}}(w)\overline{\widehat{G_{n,\nu}}(w)}
\\ & \!\!\!\!\stackrel{\eqref{eq-error-fourier-of-mollifier-2}}{\le}
C \kappa^{\delta} n^{2d} \sum_{w \in \mathbb{Z}^d_n} |\widehat{G_{n,\nu}}(w)|^2,
\end{align*}
where, $\delta = \min\{1,2\alpha\}$, as done before. In the last inequality, we
also used that $|\widehat{G_{n,\nu}}(0)|^2 \ge 0$. Since
$|F_{n,\nu}(x)|\le n^{-d}$ and due to Parseval's identity we get

\begin{align}\label{claim-bound-for-tightness-eq-3}
    \sum_{w \in \mathbb{Z}^d_n} |\widehat{G_{n,\nu}}(w)|^2
& =
    n^{-d}\sum_{x \in \mathbb{Z}^d_n} |{G_{n,\nu}}(x)|^2
\nonumber =
    n^{-d}\sum_{x \in \mathbb{T}^d_n} |{F_{n,\nu}}(x)|^2
\nonumber\\& \le
    n^{-2d}\sum_{x \in \mathbb{T}^d_n} \int_{B(x,\frac{1}{2n})} |\phi_\nu(z)|\text{d}z
 =
    n^{-2d} \|\phi_\nu(z) \|_{L^1(\mathbb{T}^d)}
 \le    Cn^{-2d}.
\end{align}
Therefore,
\begin{align}\label{claim-bound-for-tightness-eq-4}
    \sum_{x,\,y\in \mathbb{T}^d_n} \sum_{w \in \mathbb{Z}^d_n \backslash\{o\}}
    \Big( 1- \widehat{\rho_\kappa}(w) \Big)
    \frac{\exp(2 \pi i (x-y)\cdot w)}{\|w\|^{2\alpha}}
    F_{n,\nu}(x)\overline{F_{n,\nu}(y)}
    \le C \kappa^\delta.
\end{align}
We can then concentrate on bounding the first term of
\eqref{claim-bound-for-tightness-eq-2}.

\begin{align}\label{claim-bound-for-tightness-eq-5}
&
    \sum_{x,\,y\in \mathbb{T}^d_n} \sum_{w \in \mathbb{Z}^d_n \backslash\{o\}}
    \widehat{\rho_\kappa}(w) \frac{\exp(2 \pi i (x-y)\cdot w)}{\|w\|^{2\alpha}}
    F_{n,\nu}(x)\overline{F_{n,\nu}(y)}
\nonumber \\ & =
    \sum_{x,\,y\in \mathbb{T}^d_n} \sum_{w \in \mathbb{Z}^d \backslash\{o\}}
    \widehat{\rho_\kappa}(w) \frac{\exp(2 \pi i (x-y)\cdot w)}{\|w\|^{2\alpha}}
    F_{n,\nu}(x)\overline{F_{n,\nu}(y)}
\nonumber \\ & -
    \sum_{x,\,y\in \mathbb{T}^d_n}
    \sum_{\substack{w \in \mathbb{Z}^d \\ \|w\|_\infty > \frac{n}{2} }}
    \widehat{\rho_\kappa}(w) \frac{\exp(2 \pi i (x-y)\cdot w)}{\|w\|^{2\alpha}}
    F_{n,\nu}(x)\overline{F_{n,\nu}(y)}.
\end{align}
Again we use the fast decay of $\widehat{\rho_\kappa}$.
We can just apply \eqref{eq-decay-of-mollifier} for $m>d$  to get

\begin{align}\label{claim-bound-for-tightness-eq-6}
&
    \sum_{x,\,y\in \mathbb{T}^d_n}
    \sum_{\substack{w \in \mathbb{Z}^d \\ \|w\|_\infty > \frac{n}{2} }}
    \widehat{\rho_\kappa}(w) \frac{\exp(2 \pi i (x-y)\cdot w)}{\|w\|^{2\alpha}}
    F_{n,\nu}(x)\overline{F_{n,\nu}(y)}
\nonumber \\ & \le
    \frac{C}{n^{2\alpha}}
    \sum_{\substack{w \in \mathbb{Z}^d \\ \|w\|_\infty > \frac{n}{2} }}
    \widehat{\rho_\kappa}(w)
    \sum_{x,\,y\in \mathbb{T}^d_n}F_{n,\nu}(x)\overline{F_{n,\nu}(y)}
\nonumber \\ & \le
    C
    \sum_{\substack{w \in \mathbb{Z}^d \\ \|w\|_\infty > \frac{n}{2} }}
    \widehat{\rho_\kappa}(w)
    \Big|\sum_{x\in \mathbb{T}^d_n}F_{n,\nu}(x)    \Big|^2
  \le
    C
    \sum_{\substack{w \in \mathbb{Z}^d \\ \|w\|_\infty > \frac{n}{2} }}
    \frac{\| \phi_\nu\|^2_{L^1(\mathbb{T}^d)}}{(1+\|w\|)^m} \le
    C.
\end{align}

Analogously, we select $m>d$ and use \eqref{eq-decay-of-mollifier} to get
\begin{align}\label{claim-bound-for-tightness-eq-7}
    \sum_{x,\,y\in \mathbb{T}^d_n} \sum_{w \in \mathbb{Z}^d \backslash\{o\}}
&
    \widehat{\rho_\kappa}(w) \frac{\exp(2 \pi i (x-y)\cdot w)}{\|w\|^{2\alpha}}
    F_{n,\nu}(x)\overline{F_{n,\nu}(y)}
\nonumber \\ &\!\!\!\! \stackrel{\eqref{eq-decay-of-mollifier}}{\le}
    \sum_{x,\,y\in \mathbb{T}^d_n} \sum_{w \in \mathbb{Z}^d \backslash\{o\}}
    |F_{n,\nu}(x)\overline{F_{n,\nu}(y)}|
    \frac{1}{(1+\|w\|)^m}
\nonumber \\ & \le
    \sum_{w \in \mathbb{Z}^d \backslash\{o\}}
    \frac{\| \phi_\nu\|^2_{L^1(\mathbb{T}^d)}}{(1+\|w\|)^m}
 \le
    C.
\end{align}

By plugging \eqref{claim-bound-for-tightness-eq-2},
\eqref{claim-bound-for-tightness-eq-4},\eqref{claim-bound-for-tightness-eq-5},
\eqref{claim-bound-for-tightness-eq-6} and \eqref{claim-bound-for-tightness-eq-7}
in \eqref{claim-bound-for-tightness-eq-1}, we conclude the proof
of the claim, and hence of the Theorem~\ref{theorem-main-non-Gaussian} in the 
case when all moments of $\sigma$  are finite.
\end{proof}

\subsubsection{Truncation method}\label{sec-proof-truncation-method}

In the first part of the argument, we had to restrict ourselves to weights with all moments finite.
We will now show how to reconstruct the general case. We will need to fix an arbitrarily large (but finite) constant $\mathcal{R} >0$. Set

\begin{align*}
    w_n^{<\mathcal{R} }(x)&:=\sum_{y\in \mathbb{Z}_n^d}
    g^{(\alpha)}(x,\,y)\sigma(y)1_{\{|\sigma(y)|< \mathcal{R} \}},\\
    w_n^{\ge\mathcal{R} }(x)&:=\sum_{y\in \mathbb{Z}_n^d}
    g^{(\alpha)}(x,\,y)\sigma(y)1_{\{|\sigma(y)|\ge \mathcal{R} \}}.
\end{align*}
Clearly we have that
$w_n(\cdot)=w_n^{<\mathcal{R} }(\cdot)+w_n^{\ge \mathcal{R} }(\cdot)$. To
prove our result, we will use the following theorem from Theorem 4.2 from
\cite{billingsley2013convergence}.

\begin{theorem}\label{thm-from-billingsley}
    Let $S$ be a metric space with metric $\rho$. Suppose that $(X_{n,\,u},\,X_n)$ are elements of $S\times S$. If
    \[
      \lim_{u\to\infty} \overline{\lim_{n\to\infty}}
        \mathbb{P} \Big(\rho(X_{n,\,u},\,X_n)\ge \tau\Big)=0
    \]
    for all $\tau>0$, and $X_{n,\,u}\longrightarrow_{n}Z_u\longrightarrow_u X$,
    where ``$\longrightarrow_x$'' indicates convergence in law as
    $x \longrightarrow \infty$, then $X_n\longrightarrow_n X$.
\end{theorem}
Therefore, we need to prove two statements.

\begin{itemize}
\item[(S1)] $\lim_{\mathcal{R} \to\infty} \overline{\lim}_{n\to\infty} \mathbb{P} \Big(\Big\|\Xi^{\alpha}_{w_n}-\Xi^{\alpha}_{w_n^{<\mathcal{R} }}\Big\|_{\mathcal H_{-\varepsilon}}\ge \tau\Big)=0$ for all $\tau>0$.
\item[(S2)] For a constant $v_\mathcal{R} >0$, we have
    $\Xi^{\alpha}_{w_n^{<\mathcal{R} }} \longrightarrow_n
    \sqrt{v_\mathcal{R}}\Xi^{\alpha}
    \longrightarrow_{\mathcal{R} }\Xi^{\alpha}    $
    in the topology of $\mathcal H_{-\varepsilon}$.
\end{itemize}
It follows that $\Xi^{\alpha}_{w_n}$ converges to
$\Xi^{\alpha}$ in law in the topology of $\mathcal H_{-\varepsilon}$.

Since the proof of (S1) and (S2) does not present any extra technical difficulties, and therefore the argument is almost unchanged when compared to the proof in Section 5.2 in \cite{Cipriani2016}, we will leave it to the reader.

\section{Acknowledgements}
The authors would like to thank Alessandra Cipriani and Rajat Hazra for the
insightful discussions, Jan de Graaff for the pictures and Hester Kronenberg
for pointing out a mistake in a previous version of this article. Furthermore, we would like to express our gratitude to the anonymous referees who pointed out typos and mistakes in the previous version of the article and helped to improve the presentation substantially. M. Jara acknowledges
CNPq for its support through the Grant 305075/2017-9, FAPERJ for its support
through the Grant E-29/203.012/2018 and ERC for its support through the European
Unions Horizon 2020 research and innovative program (Grant Agreement No.715734).

\bibliographystyle{abbrv}
\bibliography{library}

\end{document}